\documentclass[onecolumn]{IEEEtran}
\usepackage{amsthm, amsmath, amssymb, amsfonts, url, booktabs, tikz, setspace, fancyhdr, bm, bbm}
\usepackage{multirow}
\usepackage{booktabs}
\usepackage{algorithmic}
\usepackage[ruled,linesnumbered]{algorithm2e}
\usepackage{array}
\usepackage{arydshln}
\usepackage[justification=centering,caption=false,font=normalsize,labelfont=sf,textfont=sf]{subfig}
\usepackage{textcomp}
\usepackage{stfloats}
\usepackage{comment}
\usepackage{verbatim}
\usepackage{graphicx}
\usepackage[noadjust]{cite}
\usepackage{xcolor}
\usetikzlibrary{matrix}
\usepackage{hyperref}

\newtheorem{theorem}{Theorem}
\newtheorem{definition}{Definition}
\newtheorem{lemma}{Lemma}
\newtheorem{construction}{Construction}
\newtheorem{remark}{Remark}
\newtheorem{observation}{Observation}
\newtheorem{corollary}{Corollary}

\begin{document}

\title{Criss-Cross Deletion Correcting Codes: Optimal Constructions with Efficient Decoders}

\author{Yubo~Sun and Gennian~Ge%
\thanks{This research was supported by the National Key Research and Development Program of China under Grant 2025YFC3409900, the National Natural Science Foundation of China under Grant 12231014, and Beijing Scholars Program.}
\thanks{Y. Sun ({\tt 2200502135@cnu.edu.cn}) and G. Ge ({\tt gnge@zju.edu.cn}) are with the School of Mathematical Sciences, Capital Normal University, Beijing 100048, China.}
}

\maketitle

\begin{abstract}
This paper addresses fundamental challenges in two-dimensional error correction by constructing optimal codes for \emph{criss-cross deletions}.
We consider an $ n \times n $ array $\bm{X}$ over a $ q $-ary alphabet $\Sigma_q := \{0, 1, \ldots, q-1\}$ that is subject to a \emph{$(t_r, t_c)$-criss-cross deletion}, which involves the simultaneous removal of $ t_r $ rows and $ t_c $ columns.  
A code $\mathcal{C} \subseteq \Sigma_q^{n \times n}$ is defined as a \emph{$(t_r,t_c)$-criss-cross deletion correcting code} if it can successfully correct these deletions. 

Our primary technical contributions are as follows:
\begin{itemize}
    \item \emph{Theoretical Bounds:} We derive a sphere-packing type lower bound and a Gilbert-Varshamov type upper bound on the redundancy of optimal codes.
        Our results indicate that the optimal redundancy for a $(t_r, t_c)$-criss-cross deletion correcting code lies between $(t_r + t_c)n\log q + (t_r + t_c)\log n + O_{q,t_r,t_c}(1)$ and $(t_r + t_c)n\log q + 2(t_r + t_c)\log n + O_{q,t_r,t_c}(1)$, where the logarithm is on base two, and $O_{q,t_r,t_c}(1)$ is a constant that depends solely on $q$, $t_r$, and $t_c$.

   \item \emph{Optimal Constructions:} For the case of $(1,1)$-criss-cross deletions, we propose two families of constructions that achieve $2n\log q + 2\log n + O_q(1)$ bits of redundancy. This redundancy is optimal up to an additive constant term $O_q(1)$, which depends solely on $q$. One family is designed for non-binary alphabets, while the other addresses arbitrary alphabets. For the case of $(t_r, t_c)$-criss-cross deletions, we provide a strategy to derive optimal codes when both unidirectional deletions occur consecutively. 
       
   \item \emph{Efficient Decoders:} We propose decoding algorithms with a time complexity of $O(n^2)$ for our codes, which are optimal for two-dimensional scenarios. 
\end{itemize}
\end{abstract}


\section{Introduction}
Two-dimensional error-correcting codes, where codewords are represented as two-dimensional arrays, have garnered significant attention due to their critical role in modern data storage systems.
These codes find applications in both established technologies like QR codes \cite{Abdel-Ghaffar-98, Yaakobi-10} and emerging storage paradigms including racetrack memories \cite{Lund-00-ISIT, Chee-21-ISIT}. 
While substantial progress has been made in combating substitutions and erasures for two-dimensional arrays \cite{Blaum-00, Gabidulin-85, Gabidulin-08-DCC, Roth-91, Roth-97, Sidorenko-76, Wachter-Zeh-17-IT}, 
the design of codes resilient to \emph{deletions} remains a fundamental challenge with critical implications for next-generation storage technologies.

The inherent difficulty of deletion correction persists even in one-dimensional settings, where constructing efficient deletion correcting codes remained an open problem until recent breakthroughs \cite{Brakensiek-18-IT-kD, Cai-21-E, Gabrys-23-IT-DS, Guruswami-21-IT-2D, Gabrys-19-IT-2D, Levenshtein-66-SPD-1D, Li-23-DS, Liu-24-IT, Pi-24-arXiv-2E, Nguyen-24-IT-1D, Song-22-IT, Sima-21-IT-kD, Sun-24-IT, Sun-25-IT, Sima-20-ISIT-q, Sima-20-ISIT, Sima-20-ISIT-tD, Song-22-IT-DS, Sima-20-IT-2D, Tenengolts-84-IT-q_D, Ye-24-arXiv}. 
In two-dimensional scenarios, the problem becomes significantly more complex when considering \emph{$(t_r,t_c)$-criss-cross deletions} – simultaneous deletion of $t_r$ rows and $t_c$ columns.
While unidirectional deletions (either rows or columns alone) can be reduced to one-dimensional analogs, the combined row-column deletion scenario creates an intricate combinatorial structure that resists simple dimensional reduction techniques, necessitating fundamentally new coding strategies.

Current approaches to two-dimensional deletion correction \cite{Bitar-21-IT, Chee-21-ISIT, Hagiwara-20-ISIT, Hagiwara-23, Welter-22-IT} primarily employ margin encoding strategies. 
For $(1,1)$-criss-cross deletion correcting codes, Bitar et al. \cite{Bitar-21-IT} established a sphere-packing lower bound of  $2n\log q +2\log n+ O_q(1)$ redundancy for optimal codes and constructed codes achieving a redundancy of $2n\log q +4\log n+O_q(1)$.
Subsequent work by Chee et al. \cite{Chee-21-ISIT} improved this construction to $2n + 2\log n+ O(\log\log n)$ bits for binary codes.
For general $(t_r,t_c)$-criss-cross deletion correcting codes, existing binary constructions \cite{Chee-21-ISIT, Welter-22-IT} require $(t_r+t_c)n + O_{t_r,t_c}\big((t_r+t_c)^2\log^2 n\big)$ redundancy, leaving significant room for improvement.

In this paper, motivated by both theoretical and practical perspectives, we study two-dimensional codes capable of correcting $(t_r,t_c)$-criss-cross deletions.
Firstly, we establish theoretical bounds for optimal $q$-ary $(t_r,t_c)$-criss-cross deletion correcting codes. For the lower bound, we develop a sphere-packing type bound through the analysis of arrays containing at least one row and one column with sufficiently large numbers of runs. For the upper bound, we derive a Gilbert-Varshamov type bound by enumerating arrays obtainable through the deletion of $t_r$ rows and $t_c$ columns, followed by the insertion of an equivalent number of rows and columns.
Our analysis reveals that the optimal redundancy for $(t_r, t_c)$-criss-cross deletion correcting codes ranges from 
\[(t_r + t_c)n\log q + (t_r + t_c)\log n + O_{q,t_r,t_c}(1)\]
to 
\[(t_r + t_c)n\log q + 2(t_r + t_c)\log n + O_{q,t_r,t_c}(1).\]
Then we examine the construction of $(1,1)$-criss-cross deletion correcting codes, which represents the first non-trivial case in the two-dimensional context.
Unlike previous works that employed margin encoding strategies, we adopt a strategy of full row/column encoding.
A key innovation involves defining the \emph{column composition representation} $\mathrm{CCR}(\boldsymbol{X})$, which captures the composition of each column in the array $\boldsymbol{X}$.  
Our crucial observation indicates that recovering the deleted symbol in each column inherently provides information about the column composition.
By enforcing the following two constraints:
\begin{itemize}
    \item $\mathrm{CCR}(\boldsymbol{X})$ belongs to a one-dimensional deletion correcting code,
    \item adjacent entries in $\mathrm{CCR}(\boldsymbol{X})$ maintain distinct values,
\end{itemize}
we can precisely determine the column deletion position and correct the column deletion. 
Subsequent row deletion correction is reduced to equivalent one-dimensional deletion correction. 
This method leads to the construction of non-binary codes with $2n\log q + 2\log n+O_q(1)$ bits of redundancy. 
While effective for non-binary alphabets, the adjacency constraint in $\mathrm{CCR}(\boldsymbol{X})$ necessitates a sufficiently large amount of redundancy for binary alphabets. 
Therefore, when considering binary alphabets, we relax this constraint to allow identical adjacent entries while prohibiting a limited number of consecutive identical entries. This modification introduces positional ambiguity confined to small intervals, which we resolve through:
\begin{itemize}
    \item locating the row deletion position within a small interval using similar constraints;
    \item correcting deletions within small intervals.
\end{itemize}
This refined approach yields $(1,1)$-criss-cross deletion correcting codes with $2n\log q + 2\log n + O_q(1)$ bits of redundancy for arbitrary alphabets.
Both constructions match the sphere-packing bound, indicating that they are optimal.
Finally, for the construction of $(t_r, t_c)$-criss-cross deletion correcting codes, we examine the scenario where both unidirectional deletions occur consecutively, termed \emph{$(t_r, t_c)$-criss-cross burst-deletion}, and present a strategy to derive optimal codes.

The remainder of this paper is organized as follows. Section~\ref{sec:pre} introduces the relevant notations, definitions, and important tools used throughout the paper. 
Section~\ref{sec:bound} establishes fundamental bounds on optimal $(t_r, t_c)$-criss-cross deletion correcting codes. 
Sections~\ref{sec:non-binary} and \ref{sec:binary} present constructions of $(1, 1)$-criss-cross deletion correcting codes that are optimal for non-binary and arbitrary alphabets, respectively, along with their decoding strategies.  
Section~\ref{sec:burst} suggests a strategy to construct optimal $(t_r, t_c)$-criss-cross burst-deletion correcting codes for arbitrary alphabets.
Finally, Section~\ref{sec:concl} concludes the paper and discusses open problems.


\section{Preliminaries}\label{sec:pre}

\subsection{Notations}

We now introduce the key notations used throughout this paper.

Let $[i:j]$ (where $i\leq j$) denote the integer interval $\{i,i+1,\ldots,j\}$.
For an integer $q\geq 2$, let $\Sigma_q:= [0:q-1]$ represent the $q$-ary alphabet. We denote by $\Sigma_q^n$ the set of all $q$-ary sequences (or vectors) of length $n$, and by $\Sigma_q^{n\times n}$ the set of all $q$-ary arrays of size $n\times n$.
For any sequence $\boldsymbol{x} \in \Sigma_q^n$, let $x_i$ denote its $i$-th entry where $i \in [1:n]$.
Thus, we can express $\boldsymbol{x}$ as $(x_1,x_2,\ldots,x_n)$.
The substring spanning positions $i_1$ to $i_2$ is denoted by $\boldsymbol{x}_{[i_1:i_2]}$.
For any array $\boldsymbol{X} \in \Sigma_q^{n\times n}$, let $X_{i,j}$ denote the entry at the $i$-th row and $j$-th column, where $i,j \in [1:n]$.
The full array can be represented as:
\begin{align*}
    \boldsymbol{X}= 
  \begin{bmatrix}
    X_{1,1}   & X_{1,2}  & \cdots & X_{1,n}\\
    X_{2,1}   & X_{2,2}  & \cdots & X_{2,n}\\
    \vdots    & \vdots   & \ddots & \vdots\\
    X_{n,1}   & X_{n,2}  &\cdots & X_{n,n}
  \end{bmatrix}.
\end{align*}
We use $\boldsymbol{X}_{i,[1:n]}$ and $\boldsymbol{X}_{[1:n],j}$ to denote the $i$-th row and $j$-th column of $\boldsymbol{X}$, respectively.
The subarray spanning rows $i_1$ to $i_2$ and columns $j_1$ to $j_2$ is denoted by $\boldsymbol{X}_{[i_1:i_2],[j_1:j_2]}$.
This enables the following block decomposition:
\begin{equation}\label{eq:matrix}
    \boldsymbol{X}= 
    \begin{bmatrix}
        \boldsymbol{X}_{[1:i-1],[1:j-1]}          & \boldsymbol{X}_{[1:i-1],j} & \boldsymbol{X}_{[1:i-1],[j+1:n]}\\
        \boldsymbol{X}_{i,[1:j-1]}  & X_{i,j}  & \boldsymbol{X}_{i,[j+1:n]}  \\
        \boldsymbol{X}_{[i+1:n],[1:j-1]}          & \boldsymbol{X}_{[i+1:n],j} & \boldsymbol{X}_{[i+1:n],[j+1:n]}\\
      \end{bmatrix}.
\end{equation}
The transpose of $\boldsymbol{X}$, denoted as $\boldsymbol{X}^T$, satisfies $X_{i,j}^T=X_{j,i}$, for $i,j\in [1:n]$. 

\subsection{Error Models and Codes}

We now formalize row and column deletion errors.
Let $\boldsymbol{X}$ be defined as in Equation (\ref{eq:matrix}).
\begin{itemize}
    \item A \emph{row deletion at position $i$} in $\boldsymbol{X}$ removes the $i$-th row, resulting in an $(n-1)\times n$ array:
    \begin{equation*}
        \left[
          \begin{matrix}
            \boldsymbol{X}_{[1:i-1],[1:j-1]}          & \boldsymbol{X}_{[1:i-1],j} & \boldsymbol{X}_{[1:i-1],[j+1:n]}\\
            \boldsymbol{X}_{[i+1:n],[1:j-1]}          & \boldsymbol{X}_{[i+1:n],j} & \boldsymbol{X}_{[i+1:n],[j+1:n]}\\
          \end{matrix}
      \right].
    \end{equation*}

    \item A \emph{column deletion at position $j$} in $\boldsymbol{X}$ removes the $j$-th column, resulting an $n\times (n-1)$ array:
    \begin{equation*}
        \left[
          \begin{matrix}
            \boldsymbol{X}_{[1:i-1],[1:j-1]}          & \boldsymbol{X}_{[1:i-1],[j+1:n]}\\
            \boldsymbol{X}_{i,[1:j-1]}  & \boldsymbol{X}_{i,[j+1:n]}  \\
            \boldsymbol{X}_{[i+1:n],[1:j-1]}        & \boldsymbol{X}_{[i+1:n],[j+1:n]}\\
          \end{matrix}
      \right].
    \end{equation*}

    \item A \emph{$(1,1)$-criss-cross deletion at location $(i,j)$} in $\boldsymbol{X}$ removes both the $i$-th row and $j$-th column, yielding an $(n-1) \times (n-1)$ array:
    \begin{equation*}
        \left[
          \begin{matrix}
            \boldsymbol{X}_{[1:i-1],[1:j-1]}          & \boldsymbol{X}_{[1:i-1],[j+1:n]}\\
            \boldsymbol{X}_{[i+1:n],[1:j-1]}        & \boldsymbol{X}_{[i+1:n],[j+1:n]}\\
          \end{matrix}
      \right].
    \end{equation*}
\end{itemize}

For general integers $t_r$ and $t_c$, a \emph{$(t_r, t_c)$-criss-cross deletion} in $\boldsymbol{X}$ refers to the removal of any $ t_r $ rows and $ t_c $ columns.
The collection of all possible $(n-t_r)\times (n-t_c)$ arrays resulting from such deletions is denoted by $\mathbb{D}_{t_r,t_c}(\boldsymbol{X})$, called the \emph{$(t_r, t_c)$-criss-cross
deletion ball} of $\boldsymbol{X}$.
A code $\mathcal{C}\subseteq \Sigma_q^{n\times n}$ is called a \emph{$q$-ary $(t_r, t_c)$-criss-cross deletion correcting code} if it satisfies:
\begin{align*}
    \mathbb{D}_{t_r,t_c}(\boldsymbol{X})\cap \mathbb{D}_{t_r,t_c}(\boldsymbol{Z})=\emptyset, \quad\forall \boldsymbol{X}\neq \boldsymbol{Z}\in \mathcal{C}.
\end{align*}
When deletions occur consecutively, we refer to these deletions as a \emph{burst-deletion}. 
A \emph{$(t_r, t_c)$ criss-cross burst-deletion} in $\boldsymbol{X}$ refers to the removal of any $ t_r $ consecutive rows and $ t_c $ consecutive columns, and the \emph{$(t_r, t_c)$-criss-cross
burst-deletion ball} of $\boldsymbol{X}$ is denoted by $\mathbb{D}_{t_r,t_c}^{\mathrm{burst}}(\boldsymbol{X})$.
A code $\mathcal{C}\subseteq \Sigma_q^{n\times n}$ is called a \emph{$q$-ary $(t_r, t_c)$-criss-cross burst-deletion correcting code} if it satisfies:
\begin{align*}
    \mathbb{D}_{t_r,t_c}^{\mathrm{burst}}(\boldsymbol{X})\cap \mathbb{D}_{t_r,t_c}^{\mathrm{burst}}(\boldsymbol{Z})=\emptyset, \quad\forall \boldsymbol{X}\neq \boldsymbol{Z}\in \mathcal{C}.
\end{align*}

Analogously, a \emph{$(t_r, t_c)$-criss-cross insertion} involves inserting $t_r$ rows and $t_c$ columns, while a \emph{$(t_r, t_c)$-criss-cross burst-insertion} involves inserting $t_r$ consecutive rows and $t_c$ consecutive columns. Let $\mathbb{I}_{t_r, t_c}(\boldsymbol{X})$ and $\mathbb{I}_{t_r, t_c}^{\mathrm{burst}}(\boldsymbol{X})$ be the \emph{$(t_r, t_c)$-criss-cross insertion ball} and \emph{$(t_r, t_c)$-criss-cross burst-insertion ball} of $\boldsymbol{X}$, respectively. A code $\mathcal{C} \subseteq \Sigma_q^{n \times n}$ is a \emph{$q$-ary $(t_r, t_c)$-criss-cross insertion correcting code} (respectively, a \emph{$q$-ary $(t_r, t_c)$-criss-cross burst-insertion correcting code}) if it can uniquely recover any array affected by the corresponding types of insertions.  

To evaluate an error-correcting code $\mathcal{C}\subseteq \Sigma_q^{n\times n}$, we calculate its \emph{redundancy}, defined as $n^2 \log q- \log |\mathcal{C}|$, where $\log$ denotes base-$2$ logarithm and $|\mathcal{C}|$ is the cardinality of the code.

\begin{remark}
The order of row deletions and column deletions is commutative in terms of the final result. Therefore, without loss of generality, we adopt the convention that row deletions precede column deletions.
\end{remark}

Bitar et al. \cite{Bitar-21-IT} and Welter et al. \cite{Welter-22-IT} established foundational equivalences between criss-cross deletion correcting codes and their insertion-correcting counterparts.

\begin{lemma}\cite[Corollary 2]{Bitar-21-IT}
    Let $t_r=t_c=t$. A code $\mathcal{C}\subseteq \Sigma_q^{n\times n}$ is a $(t,t)$-criss-cross deletion correcting code if and only if it is a $(t,t)$-criss-cross insertion correcting code.
\end{lemma}

Subsequently, Stylianou et al. \cite{Stylianou-22-ISIT} extended this equivalence to a broader framework, which we formalize below.

\begin{lemma}\cite[Lemma 3]{Stylianou-22-ISIT}\label{lem:equiv}
  A code $\mathcal{C}\subseteq \Sigma_q^{n \times n}$ can correct $t_r$-row deletions and $t_c$-column deletions if and only if it can correct $t_r$-row insdels and $t_c$-column insdels, referred to as \emph{$(t_r,t_c)$-criss-cross insdels}. Here, the term ``$t$ insdels'' denotes any combination of $t^{ins}$ insertions and $t^{del}$ deletions satisfying $t^{ins}+t^{del}=t$.
\end{lemma}

Building upon the well-established foundation of one-dimensional burst-error correction theory \cite{Schoeny-17-IT-BD}, we formalize the foundational equivalence between two-dimensional burst-deletion correction capabilities and their corresponding burst-insertion counterparts. For clarity of exposition, we first recall critical results from the one-dimensional setting:

\begin{lemma}\cite[Theorem 1]{Schoeny-17-IT-BD}\label{lem:burst}
  A code $\mathcal{C}\subseteq \Sigma_q^{n}$ can correct $t$ consecutive deletions if and only if it can correct $t$ consecutive insertions.
\end{lemma}

\begin{lemma}\label{lem:burst-equiv}
  A code $\mathcal{C}\subseteq \Sigma_q^{n\times n}$ is a $(t_r,t_c)$-criss-cross burst-deletion correcting code if and only if it is a $(t_r,t_c)$-criss-cross burst-insertion correcting code.
\end{lemma}

\begin{IEEEproof}
    We prove it by contradiction through contrapositive analysis.
    Suppose $\mathcal{C}$ fails to be a $(t_r,t_c)$-criss-cross burst-deletion correcting code. Then there exist distinct codewords $\boldsymbol{X}, \boldsymbol{Y}\in \mathcal{C}$ with overlapping $(t_r,t_c)$-criss-cross burst-deletion balls, i.e., some array $\boldsymbol{Z}$ satisfies
    \begin{gather*}
      \boldsymbol{X}\xrightarrow{t_r \text{ consecutive row deletions}} \boldsymbol{X}' \xrightarrow{t_c \text{ consecutive column deletions}} \boldsymbol{Z},\\
      \boldsymbol{Y}\xrightarrow{t_r \text{ consecutive row deletions}} \boldsymbol{Y}' \xrightarrow{t_c \text{ consecutive column deletions}} \boldsymbol{Z}.
    \end{gather*}
    This yields $\mathbb{D}_{0,t_c}(\boldsymbol{X}')\cap \mathbb{D}_{0,t_c}(\boldsymbol{Y}')\neq \emptyset$. By Lemma \ref{lem:burst}, there exists $\boldsymbol{Z}'\in \mathbb{I}_{0,t_c}(\boldsymbol{X}')\cap \mathbb{I}_{0,t_c}(\boldsymbol{Y}')$ such that:
    \begin{gather*}
      \boldsymbol{X}\xrightarrow{t_r \text{ consecutive row deletions}} \boldsymbol{X}' \xrightarrow{t_c \text{ consecutive column insertions}} \boldsymbol{Z}',\\
      \boldsymbol{Y}\xrightarrow{t_r \text{ consecutive row deletions}} \boldsymbol{Y}' \xrightarrow{t_c \text{ consecutive column insertions}} \boldsymbol{Z}'.
    \end{gather*}
    By changing the order of row errors and column errors, we obtain intermediate arrays $\boldsymbol{X}'', \boldsymbol{Y}''$ satisfying
    \begin{gather*}
      \boldsymbol{X}\xrightarrow{t_c \text{ consecutive column insertions}} \boldsymbol{X}'' \xrightarrow{t_r \text{ consecutive row deletions}} \boldsymbol{Z}',\\
      \boldsymbol{Y}\xrightarrow{t_c \text{ consecutive column insertions}} \boldsymbol{Y}'' \xrightarrow{t_r \text{ consecutive row deletions}} \boldsymbol{Z}'.
    \end{gather*}
    This implies $\mathbb{D}_{t_r,0}(\boldsymbol{X}'')\cap \mathbb{D}_{t_r,0}(\boldsymbol{Y}'')\neq \emptyset$. Applying Lemma \ref{lem:burst} again, we find $\boldsymbol{Z}''\in \mathbb{I}_{t_r,0}(\boldsymbol{X}'')\cap \mathbb{I}_{t_r,0}(\boldsymbol{Y}'')$ such that
    \begin{gather*}
      \boldsymbol{X}\xrightarrow{t_c \text{ consecutive column insertions}} \boldsymbol{X}'' \xrightarrow{t_r \text{ consecutive row insertions}} \boldsymbol{Z}'',\\
      \boldsymbol{Y}\xrightarrow{t_c \text{ consecutive column insertions}} \boldsymbol{Y}'' \xrightarrow{t_r \text{ consecutive row insertions}} \boldsymbol{Z}''.
    \end{gather*}
    Thus, $\mathcal{C}$ cannot be a $(t_r,t_c)$-criss-cross burst-insertion correcting code.
    The converse follows symmetrically, thereby completing the proof.
\end{IEEEproof}

\begin{remark}
  A burst of $t$ insdels refers to either a burst of $t$ insertions or a burst of $t$ deletions.
  We say that $\boldsymbol{X}\in \Sigma_q^{n\times n}$ experiences a \emph{$(t_r,t_c)$-criss-cross burst-insdel} if it suffers a burst of $t_r$-row insdels and a burst of $t_c$-column insdels.
  One can follow an analogous discussion found in Lemma \ref{lem:burst-equiv} to demonstrate that a code $\mathcal{C}\subseteq \Sigma_q^{n\times n}$ can correct a $(t_r,t_c)$-criss-cross burst-deletion if and only if it can correct a $(t_r,t_c)$-criss-cross burst-insdel.
\end{remark}

\subsection{Useful Tools}

We now review a classical construction of one-dimensional non-binary single deletion correcting codes, which will be crucial for our subsequent developments of two-dimensional criss-cross deletion correcting codes.

\begin{definition}
    For any sequence $\boldsymbol{x} = (x_1,x_2,\ldots,x_n)\in \Sigma_q^n$, its \emph{signature} is defined as the binary vector $\alpha(\boldsymbol{x})\in \Sigma_2^{n-1}$, where its $i$-th entry
    \begin{equation*}
        \alpha(x)_i=
        \begin{cases}
            1, &\mbox{if }x_{i+1} \geq x_i;\\
            0, &\mbox{otherwise},
        \end{cases}
    \end{equation*}
    for $i \in [1:n - 1]$. 
\end{definition}

\begin{lemma}\cite[Theorem 1]{Tenengolts-84-IT-q_D}
    For any $a \in \Sigma_n$ and $b \in \Sigma_q$, the code 
    \begin{align*}
        \mathrm{VT}_{a,b}(n,q) = \left\{\boldsymbol{x} \in \Sigma_q^n : \sum_{i=1}^{n-1} i\cdot \alpha(x)_i \equiv a \pmod{n},
        ~ \sum_{i=1}^n x_i \equiv b \pmod{q} \right\}
    \end{align*} 
    forms a $q$-ary single deletion correcting code of length $n$.
\end{lemma}

This code family $\mathrm{VT}_{a,b}(n,q)$ commonly referred to as the \emph{non-binary Varshamov-Tenengolts code} (abbreviated as \emph{non-binary VT code}), admits an efficient decoding algorithm, which is outlined as follows:
Let $\boldsymbol{x}\in \mathrm{VT}_{a,b}(n,q)$ experience a single deletion at position $d$, resulting in the subsequence $\boldsymbol{y}$.
\begin{itemize}
    \item \emph{Symbol Recovery:} Compute $\Delta:= \left(d-\sum_{i=1}^{n-1} y_i\right) \pmod{q}$. The missing symbol is uniquely determined as $x_d= \Delta$.

    \item \emph{Signature Recovery:} Let $\alpha(\boldsymbol{y})\in \Sigma_2^{n-2}$ denote the signature of $\boldsymbol{y}$.  The original signature $\alpha(\boldsymbol{x})$ can be uniquely reconstructed by using $\alpha(\boldsymbol{y})$ and the condition $\sum_{i=1}^{n-1} i\cdot \alpha(x)_i \equiv a \pmod{n}$.

    \item \emph{Sequence Recovery:} Compare $\alpha(\boldsymbol{y})$ with $\alpha(\boldsymbol{x})$.
    A `unique' position $d'\in [1:n]$ is identified where inserting $x_d$ into $\boldsymbol{y}$ restores consistency with $\alpha(\boldsymbol{x})$.
\end{itemize}
This process guarantees unique reconstruction of $\boldsymbol{x}$, leading to:

\begin{corollary}\label{cor:VTdecoder}
    For any $a \in \Sigma_n$, the code
    \begin{align*}
        \left\{\boldsymbol{x} \in \Sigma_q^n : \sum_{i=1}^{n-1} i\cdot \alpha(x)_i \equiv a \pmod{n}\right\}
    \end{align*}
    can uniquely correct any single deletion when the deleted symbol is known a priori.
\end{corollary}

A code $\mathcal{C} \subseteq \Sigma_q^n$ is called a \emph{$P$-bounded single deletion correcting code} if it can correct a single deletion given the additional knowledge of an interval of length $P$ containing the erroneous coordinate. 

\begin{lemma}\cite[Lemma 1]{Schoeny-17-SVT}
  For any $a \in \Sigma_{P+1}$, $b \in \Sigma_q$, and $c\in \Sigma_2$, the code 
    \begin{align*}
        \mathrm{SVT}_{a,b,c}(n,q,P) = \left\{\boldsymbol{x} \in \Sigma_q^n : \sum_{i=1}^{n-1} i\cdot \alpha(x)_i \equiv a \pmod{P+1},
        ~ \sum_{i=1}^n x_i \equiv b \pmod{q},~ \sum_{i=1}^n \alpha(x)_i \equiv c \pmod{2}\right\}
    \end{align*} 
    forms a $P$-bounded single deletion correcting code.
\end{lemma}

This code family $\mathrm{SVT}_{a,b,c}(n,q,P)$ commonly referred to as the \emph{non-binary $P$-bounded VT code}, also admits an efficient decoding algorithm.
Similar to Corollary \ref{cor:VTdecoder}, we can derive the following:
\begin{corollary}\label{cor:sVTdecoder}
    For any $a \in \Sigma_{P+1}$ and $c\in \Sigma_2$, the code
    \begin{align*}
        \left\{\boldsymbol{x} \in \Sigma_q^n : \sum_{i=1}^{n-1} i\cdot \alpha(x)_i \equiv a \pmod{P+1},~ \sum_{i=1}^n \alpha(x)_i \equiv c \pmod{2}\right\}
    \end{align*}
    can uniquely correct any single deletion given the deleted symbol and an interval of length $P$ containing the erroneous coordinate.
\end{corollary}

For the special case where $P=2$, we will now introduce a construction that achieves the same objective as in Corollary \ref{cor:sVTdecoder} using only one bit of redundancy. 
The number of \emph{inversions} in a sequence $\boldsymbol{x}\in \Sigma_q^n$ is defined as $\mathrm{Inv}(\boldsymbol{u}):= \left| \{(s,t) : 1 \leq s < t \leq n, x_s > x_t \} \right|$. As established by \cite[Lemma 6.9]{Sun-23-IT-BDR}, inserting a symbol into two consecutive positions of $\bm{x}$ results in two new vectors whose inversions have distinct parities. This leads directly to the following result.

\begin{corollary}\label{cor:inversion}
    For any $c\in \Sigma_2$, the code
    \begin{align*}
        \left\{\boldsymbol{x} \in \Sigma_q^n : \mathrm{Inv}(\bm{x}) \equiv c \pmod{2}\right\}
    \end{align*}
    can uniquely correct any single deletion given the deleted symbol and an interval of length two containing the erroneous coordinate.
\end{corollary}

\section{Bounds on the Redundancy of Optimal \texorpdfstring{$(t_r, t_c)$}{}-Criss-Cross Deletion Correcting Codes}\label{sec:bound}

To the best of our knowledge, there are no existing upper or lower bounds on the redundancy of optimal $(t_r, t_c)$-criss-cross deletion correcting codes, except for a sphere-packing type lower bound specific to the case where $(t_r, t_c) = (1,1)$. In \cite{Bitar-21-IT}, Bitar et al. derived a sphere-packing type lower bound on the redundancy of $(1, 1)$-criss-cross deletion correcting codes by characterizing the arrays that can be obtained from a $(1, 1)$-criss-cross deletion. Their method can be viewed as a generalization of the approach developed for the one-dimensional case. However, their discussion is complex, making it challenging to generalize their method to accommodate $(t_r, t_c)$-criss-cross deletions.

In Subsection \ref{subsec:SP-bound}, we propose a simpler approach to derive a sphere-packing type lower bound on the redundancy of $(t_r, t_c)$-criss-cross deletion correcting codes. 
Our key idea consists of the following two steps:
\begin{itemize}
    \item Demonstrating that for any $ n \times n $ array with at least one row and one column having a number of runs, where a run in a sequence is defined as a maximal substring consisting of the same symbol, that is asymptotically lower bounded by $\frac{q-1}{q}n$, the size of its $(t_r, t_c)$-criss-cross deletion ball is on the order of $ n^{t_r+t_c} $;
    \item Establishing that the number of $ n \times n $ arrays whose rows and columns contain a number of runs asymptotically below $\frac{q-1}{q}n$ is significantly less than $\frac{q^{n^2-(t_r+t_c)n}}{n^{t_r+t_c}}$.
\end{itemize}
Although our method is straightforward, the resulting conclusion is almost as good as that of \cite{Bitar-21-IT}, even for the parameters $(t_r, t_c) = (1, 1)$.
Subsequently, in Subsection \ref{subsec:GV-bound} we present a Gilbert-Varshamov type upper bound on the redundancy of optimal $(t_r, t_c)$-criss-cross deletion correcting codes by demonstrating that for a given $ n \times n $ array $\boldsymbol{X}$, the number of arrays whose $(t_r, t_c)$-criss-cross deletion ball intersects non-trivially with the $(t_r, t_c)$-criss-cross deletion ball of $\boldsymbol{X}$ is on the order of $ q^{(t_r+t_c)n} n^{2(t_r+t_c)} $.

\subsection{Sphere-Packing Type Lower Bound}\label{subsec:SP-bound}

\begin{theorem}\label{thm:lb}
    Let $\epsilon = \sqrt{\frac{(t_r+t_c+1) \ln n}{2(n-1)}}$ be such that $0<\epsilon- \frac{1/q+\epsilon-4\max\{t_r,t_c\}+4}{n}\leq \frac{q-1}{2q}$. When $n\geq q$, the optimal redundancy of $q$-ary $(t_r, t_c)$-criss-cross deletion correcting codes is lower bounded by  $(t_r+t_c)n \log q+ (t_r+t_c)\log n+O_{q,t_r,t_c}(1)$.
\end{theorem}

\begin{proof}
Let $\mathcal{C} \subseteq \Sigma_q^{n \times n}$ be a $(t_r, t_c)$-criss-cross deletion correcting code. 
Set $t = \big(\frac{q-1}{q} - \epsilon\big)(n - 1)+1$.
We partition $\mathcal{C}$ into two disjoint subsets: $\mathcal{C} = \mathcal{C}_1 \sqcup \mathcal{C}_2$. The subset $\mathcal{C}_1$ consists of those arrays in $\mathcal{C}$ that have at least one row and one column containing a significant number of runs, while $\mathcal{C}_2$ includes the remaining arrays in $\mathcal{C}$.
Specifically, let $\gamma(\boldsymbol{x})$ denote the number of runs in sequence $\boldsymbol{x}$, we define:
\begin{align*}
\mathcal{C}_1 
= \big\{\boldsymbol{X} \in \mathcal{C}: &~\gamma(\boldsymbol{X}_{i,[1:n]}) \geq t+1 \text{ for some } i\in [1:n],\\
&~~~~\gamma(\boldsymbol{X}_{[1:n],j}) \geq t + 1 \text{ for some } j\in [1:n]\big\}
\end{align*}
and 
\begin{align*}
\mathcal{C}_2 
= \big\{\boldsymbol{X} \in \mathcal{C}: \gamma(\boldsymbol{X}_{i,[1:n]})\leq t, \gamma(\boldsymbol{X}_{[1:n],i}) \leq t \text{ for } i\in [1:n]\big\}.
\end{align*}

For any $\boldsymbol{X} \in \mathcal{C}_1$, there exist indices $i,j\in [1:n]$ such that $\gamma(\boldsymbol{X}_{i,[1:n]}) \geq t+1$ and $\gamma(\boldsymbol{X}_{[1:n],j}) \geq t+1$.
Observe that $\gamma(\boldsymbol{X}_{i,[1:j-1]})+ \gamma(X_{i,j})+ \gamma(\boldsymbol{X}_{i,[j+1:n]})\geq \gamma(\boldsymbol{X}_{i,[1:n]})\geq t+1$.
This implies that $\gamma(\boldsymbol{X}_{i,[1:j-1]})+ \gamma(\boldsymbol{X}_{i,[j+1:n]})\geq t$, which, by the pigeonhole principle, leads us to conclude that either $\gamma(\boldsymbol{X}_{i,[1:j-1]})\geq t/2$ or $\gamma(\boldsymbol{X}_{i,[j+1:n]})\geq t/2$.
Similarly, we can also deduce that either $\gamma(\boldsymbol{X}_{[1:i-1],j})\geq t/2$ or $\gamma(\boldsymbol{X}_{[i+1:n],j})\geq t/2$.
Without loss of generality, we assume $\gamma(\boldsymbol{X}_{i,[1:j-1]})\geq t/2$ and $\gamma(\boldsymbol{X}_{[1:i-1],j})\geq t/2$, as the other cases can be discussed similarly.
Now, we define $\mathbb{D}_{t_r, t_c}'(\boldsymbol{X}) \subseteq \Sigma_q^{(n-t_r)\times (n-t_c)}$ as the set of arrays that can be obtained from $\boldsymbol{X}$ by deleting $t_r$ rows from its first $i-1$ rows and $t_c$ columns from its first $j-1$ columns. 
Clearly, $\mathbb{D}_{t_r, t_c}'(\boldsymbol{X})$ is a subset of $\mathbb{D}_{t_r, t_c}(\boldsymbol{X})$.
By counting the number of choices for the $(i-t_r)$-th row and $(j-t_c)$-th column in $\mathbb{D}_{t_r, t_c}'(\boldsymbol{X})$, we can derive a lower bound on the size of $\mathbb{D}_{t_r, t_c}(\boldsymbol{X})$:
\begin{align*}
  |\mathbb{D}_{t_r, t_c}(\boldsymbol{X})|
  &\geq |\mathbb{D}_{t_r, t_c}'(\boldsymbol{X})|\\
  &\geq |\mathbb{D}_{t_r}(\boldsymbol{X}_{i,[1:j-1]})| \cdot |\mathbb{D}_{t_c}(\boldsymbol{X}_{[1:i-1],j})|\\
  &\stackrel{(\star)}{\geq} \binom{\frac{t}{2}-t_r+1}{t_r}\binom{\frac{t}{2}-t_c+1}{t_c},
\end{align*}
where the last inequality can be justified by \cite[Equation (1)]{Levenshtein-66-SPD-1D}
 that the size of an $s$-deletion ball centered at a sequence with at least $r$ runs is lower bounded by $\binom{r-s+1}{s}$.
Since $\mathcal{C}$ is a $(t_r, t_c)$-criss-cross deletion correcting code and $\mathcal{C}_1$ is a subset of $\mathcal{C}$, we know that $\mathcal{C}_1$ is also a $(t_r, t_c)$-criss-cross deletion correcting code.
Then we can derive the following equation:
\begin{align*}
  \Big| \bigcup_{\boldsymbol{X} \in \mathcal{C}_1} \mathbb{D}_{t_r,t_c}(\boldsymbol{X}) \Big|
  &= \sum_{\boldsymbol{X} \in \mathcal{C}_1}|\mathbb{D}_{t_r,t_c}(\boldsymbol{X})|\\
  &\geq |\mathcal{C}_1| \cdot \binom{\frac{t}{2}-t_r+1}{t_r}\binom{\frac{t}{2}-t_c+1}{t_c}.
\end{align*}
Since $\mathbb{D}_{t_r,t_c}(\boldsymbol{X})$ is a subset of $\Sigma_q^{(n-t_r)\times (n-t_c)}$ for $\boldsymbol{X}\in \mathcal{C}_1$, we have
\begin{align*}
  \Big| \bigcup_{\boldsymbol{X} \in \mathcal{C}_1} \mathbb{D}_{t_r,t_c}(\boldsymbol{X}) \Big| \leq q^{(n-t_r)(n-t_c)}.
\end{align*}
This implies that 
\begin{align*}
  |\mathcal{C}_1| 
  &\leq \frac{q^{(n-t_r)(n-t_c)}}{\binom{\frac{t}{2}-t_r+1}{t_r}\binom{\frac{t}{2}-t_c+1}{t_c}} \\
  &\leq \frac{t_r!t_c!q^{(n-t_r)(n-t_c)}}{(\frac{t}{2}-2t_r+2)^{t_r}(\frac{t}{2}-2t_c+2)^{t_c}}\\
  &= \frac{t_r!t_c!q^{(n-t_r)(n-t_c)}}{\left(K_1 \frac{(q-1)n}{2q}\right)^{t_r}\cdot \left(K_2 \frac{(q-1)n}{2q}\right)^{t_c}}\\
  &\leq \frac{4^{t_r+t_c} t_r!t_c! q^{n^2-(t_r+t_c)n+t_rt_c}}{\left(\frac{q-1}{q}\right)^{t_r+t_c}n^{t_r+t_c}}\\
  &= O_{q,t_r,t_c}\left(\frac{q^{n^2}}{n^{t_r+t_c}q^{(t_r+t_c)n}}\right),
\end{align*}
where
\begin{gather*}
    K_1= 1-\frac{q}{q-1}\left(\epsilon- \frac{1/q+\epsilon-4t_r+4}{n} \right)\geq \frac{1}{2},\\
    K_2= 1-\frac{q}{q-1}\left(\epsilon- \frac{1/q+\epsilon-4t_c+4}{n} \right)\geq \frac{1}{2}.
\end{gather*}

Next, we proceed to upper bound $|\mathcal{C}_2|$. To this end, we define the following set:
\begin{align*}
\mathcal{C}_2 '
= \big\{\boldsymbol{X} \in \mathcal{C}: r(\boldsymbol{X}_{i,[1:n]}) \leq t \text{ for each $i\in [1,n]$}\big\}.
\end{align*}
It is clear that $\mathcal{C}_2 \subseteq \mathcal{C}_2'$ and that $|\mathcal{C}_2|\leq |\mathcal{C}_2'|$.
By \cite[Lemma 10]{Cullina-14-IT}, the number of $q$-ary sequences of length $n$ containing $t=(\frac{q-1}{q}-\epsilon)(n-1)+1$ or fewer runs is upper bounded by $q^{n}e^{-2(n-1)\epsilon^2}$.
Recall that $\epsilon = \sqrt{\frac{(t_r+t_c+1) \ln n}{2(n-1)}}$, we can compute
\begin{align*}
  |\mathcal{C}_2|
  \leq|\mathcal{C}_2'|
  \leq \big(q^ne^{-2(n-1)\epsilon^2}\big)^n
  = \frac{q^{n^2}}{n^{n(t_r+t_c+1)}}\leq \frac{q^{n^2}}{n^{t_r+t_c}q^{(t_r+t_c)n}},
\end{align*}
where the last inequality holds by $n\geq q$.

Consequently, we can calculate
\begin{align*}
  |\mathcal{C}|
  &= |\mathcal{C}_1|+ |\mathcal{C}_2|\leq O_{q,t_r,t_c}\left(\frac{q^{n^2}}{n^{t_r+t_c}q^{(t_r+t_c)n}}\right).
\end{align*}
The conclusion follows from the definition of redundancy.  
\end{proof}

\begin{remark}
  Let $\epsilon = \sqrt{\frac{(t_r+t_c+1) \ln n}{2(n-1)}}$, for sufficiently large $n$, we have $\epsilon=K_1\sqrt{\frac{\ln n}{n}}$ and $\frac{1/q+\epsilon-4\max\{t_r,t_c\}+4}{n}=\frac{K_2}{n}$ for some positive constants $K_1$ and $K_2$.
  In this setting, the condition $0<\epsilon- \frac{1/q+\epsilon-4\max\{t_r,t_c\}+4}{n}\leq \frac{q-1}{2q}$ is automatically satisfied.
\end{remark}

\begin{remark}\label{rmk:burst}
  A $t$-run in a sequence is defined as a maximal substring of period $t$. For $(t_r, t_c)$-criss-cross burst-deletion correcting codes, by utilizing \cite[Claim 3.1]{Sun-23-IT-BDR} and \cite[Claim 4]{Wang-24-IT} and considering arrays that have at least one row containing a significant number of $t_r$-runs and one column containing a significant number of $t_c$-runs, one can follow an analogous discussion found in Theorem \ref{thm:lb} to demonstrate that the optimal redundancy is lower bounded by $(t_r + t_c)n \log q + 2\log n + O_{q, t_r, t_c}(1)$. 
\end{remark}

\subsection{Gilbert-Varshamov Type Upper Bound}\label{subsec:GV-bound}

\begin{theorem}\label{thm:ub}
    The optimal redundancy of $q$-ary $(t_r, t_c)$-criss-cross deletion correcting codes is upper bounded by  $(t_r+t_c) n\log q+ 2(t_r+t_c) \log n+O_{q,t_r,t_c}(1)$.
\end{theorem}

\begin{proof}
    Consider a graph $G$ whose vertex set $V(G)$ is $\Sigma_q^{n\times n}$. Two distinct vertices $\boldsymbol{X}$ and $\boldsymbol{Z}$ in $V(G)$ are connected by an edge if and only if the intersection between their $(t_r,t_c)$-criss-cross deletion balls is non-empty. An independent set of $G$ is a subset of $V(G)$ such that no two distinct vertices are connected by an edge and the independence number of $G$ is the cardinality of the largest independent set of $G$. Let $\vartheta(G)$ represent the independence number of $G$. 
    By definition, a subset $\mathcal{C} \subseteq \Sigma_q^{n\times n}$ is a $(t_r,t_c)$-criss-cross deletion correcting code if and only if $\mathcal{C}$ is an independent set in $G$. Therefore, the largest size of $q$-ary $(t_r, t_c)$-criss-cross deletion correcting codes equals $\vartheta(G)$, implying the optimal redundancy of $q$-ary $(t_r, t_c)$-criss-cross deletion correcting codes is $n^2\log q- \log \big(\vartheta(G)\big)$.
    It remains to present a lower bound on the size of $\vartheta(G)$.

    For a vertex $\boldsymbol{X}$, let $d(\boldsymbol{X})$ be the number of vertices $\boldsymbol{Z}$ such that they are connected in $G$. By \cite[Page 100, Theorem 1]{Alon}, we have that 
    \begin{equation}\label{eq:alpha(G)}
         \vartheta(G) \geq \sum_{\boldsymbol{X} \in \Sigma_q^n} \frac{1}{d(\boldsymbol{X}) + 1}.
    \end{equation}
    
    Observe that each array $\boldsymbol{Z}$ (including $\boldsymbol{X}$ itself) for which the intersection between its $(t_r, t_c)$-criss-cross deletion ball and the $(t_r, t_c)$-criss-cross deletion ball of $\boldsymbol{X}$ is non-empty can be obtained through the following steps:
    \begin{enumerate}
        \item \emph{Deletion of Rows and Columns:} Delete $t_r$ rows and $t_c$ columns from $\boldsymbol{X}$, resulting in an array $\boldsymbol{Z}'$ of size $(n - t_r) \times (n - t_c)$. Since there are at most $\binom{n}{t_r} \binom{n}{t_c}$ ways to choose the rows and columns, the number of choices for $\boldsymbol{Z}'$ is upper bounded by  $\binom{n}{t_r} \binom{n}{t_c}$.
        
        \item \emph{Insertion of Column Vectors:} Insert $t_c$ column vectors of length $n - t_r$ from $\Sigma_q$ into $\boldsymbol{Z}'$ to obtain an array $\boldsymbol{Z}''$ of size $(n - t_r) \times n$ over $\Sigma_q$. 
        Since for each given $\boldsymbol{Z}'$, this operation is equivalent to inserting $t_c$ symbols from $\Sigma_{q^{n-t_r}}$ into a row vector of length $n - t_c$ over $\Sigma_{q^{n-t_r}}$.
        By \cite[Equation (24)]{Levenshtein-01-JCTA-recons}, the total number of choices for $\boldsymbol{Z}''$ can be upper bounded by 
        \[
        \binom{n}{t_r} \binom{n}{t_c} \sum_{i=0}^{t_c} \binom{n}{i} (q^{n-t_r} - 1)^i.
        \]
            
        \item \emph{Insertion of Row Vectors:} Insert $t_r$ row vectors of length $n$ from $\Sigma_q$ into $\boldsymbol{Z}''$ to obtain an array $\boldsymbol{Z}$ of size $n \times n$ over $\Sigma_q$ that is connected to $\boldsymbol{X}$ in $G$. 
        Since for each given $\boldsymbol{Z}''$, this process can be interpreted as inserting $t_r$ symbols from $\Sigma_{q^n}$ into a column vector of length $n - t_r$ over $\Sigma_{q^n}$.
        Again by \cite[Equation (24)]{Levenshtein-01-JCTA-recons}, the total number of choices for $\boldsymbol{Z}$ can be upper bounded by
        \[
        \binom{n}{t_r} \binom{n}{t_c} \sum_{i=0}^{t_c} \binom{n}{i} (q^{n-t_r} - 1)^i \sum_{j=0}^{t_r} \binom{n}{j} (q^n - 1)^j.
        \]
    \end{enumerate}
    Consequently, we get
    \begin{align*}
        d(\boldsymbol{X}) + 1 
        &\leq \binom{n}{t_r} \binom{n}{t_c} \sum_{i=0}^{t_c}\binom{n}{i}(q^{n-t_r}-1)^i \sum_{j=0}^{t_r}\binom{n}{j}(q^{n}-1)^j\\
        &\leq \frac{n^{t_r}}{t_r!} \frac{n^{t_c}}{t_c!} \frac{t_c n^{t_c} q^{t_c(n-t_r)}}{t_c!} \frac{t_r n^{t_r} q^{t_rn}}{t_r!}\\
        &= O_{q,t_r,t_c}\left(n^{2(t_r+t_c)q^{(t_r+t_c)n}}\right). 
    \end{align*}
    Then by Equation (\ref{eq:alpha(G)}), we can compute
    \begin{align*}
      \vartheta(G)
      &\geq \frac{q^{n^2}}{d(\boldsymbol{X}) + 1}= O_{q,t_r,t_c}\left(\frac{q^{n^2}}{n^{2(t_r+t_c)}q^{(t_r+t_c)n}}\right).
    \end{align*}
    Thus, the conclusion follows from the definition of redundancy.  
\end{proof}

The following corollary summarizes the main contributions of this section.

\begin{corollary}
  Let $\mathcal{C}\subseteq \Sigma_q^{n\times n}$ be an optimal $q$-ary $(t_r, t_c)$-criss-cross deletion correcting code, then its redundancy ranges from $(t_r+t_c)n \log q+ (t_r+t_c)\log n+O_{q,t_r,t_c}(1)$ to $(t_r+t_c)n\log q + 2(t_r+t_c)\log n+O_{q,t_r,t_c}(1)$.
\end{corollary}

\section{Construction of Optimal \texorpdfstring{$(1,1)$}{}-Criss-Cross deletion correcting Codes Over Non-Binary Alphabets}\label{sec:non-binary}%

This section details a construction of $(1,1)$-criss-cross deletion correcting codes that exhibits optimal redundancy (up to an additive constant term) in the context of non-binary alphabets.
We begin by discussing the underlying ideas that inform our code construction. 
Let $\boldsymbol{X}$ be an array of size $n\times n$, and let $\boldsymbol{Y}$ be the array of size $(n-1)\times (n-1)$ that is obtained by deleting the $i$-th row and $j$-th column from $\boldsymbol{X}$.
As illustrated in Figure \ref{fig:array}, we can visualize the process of obtaining $\boldsymbol{Y}$ as follows: First, we move the $i$-th row of $\boldsymbol{X}$ to the bottom, resulting in the array $\boldsymbol{X}'$. Next, we shift the $j$-th column of $\boldsymbol{X}'$ to the rightmost position, yielding $\boldsymbol{X}''$. Finally, $\boldsymbol{Y}$ is the top-left corner subarray of $\boldsymbol{X}''$, with size $(n-1) \times (n-1)$.
Our approach involves the following three key steps:
\begin{itemize}
  \item \emph{Determine the missing symbols in each row and column:} We enforce that the sums of the entries in each row and column of $\boldsymbol{X}$ yield the same residual modulo $q$. This ensures that if the array experiences a row deletion followed by a column deletion, the deleted symbols within each row and column can be uniquely determined.
  In other words, this allows us to determine $\boldsymbol{X}''$ from $\boldsymbol{Y}$.
  
  \item \emph{Determine the column deletion position $j$:} By knowing $\boldsymbol{X}''$, we can ascertain the composition of each column of $\boldsymbol{X}$. The \emph{composition} of a sequence $\boldsymbol{u}\in \Sigma_q^n$ is represented by
       \begin{align*}
        \mathrm{Comp}(\boldsymbol{u}):= (\boldsymbol{u}|_0, \boldsymbol{u}|_1,\ldots, \boldsymbol{u}|_{q-1}),
      \end{align*}
      where $\boldsymbol{u}|_{i}$ denotes the number of occurrences of $i$ in $\bm{u}$, for $i\in [0:q-1]$.
      Clearly, the compositions of each column of $\boldsymbol{X}$ form a vector, denoted by $\mathrm{CCR}(\boldsymbol{X})$, of length $n$ over an alphabet of size $\binom{n+q-1}{q-1}$. We require that $\mathrm{CCR}(\boldsymbol{X})$ belongs to some one-dimensional single deletion correcting code, satisfying the additional constraint that adjacent entries are distinct.   
        Observe that $\mathrm{CCR}(\boldsymbol{X}'')_{[1:n-1]}$ can be obtained from $\mathrm{CCR}(\boldsymbol{X})$ by deleting its $j$-th entry $\mathrm{CCR}(\boldsymbol{X}'')_{n}$, allowing us to recover $\mathrm{CCR}(\boldsymbol{X})$ from $\mathrm{CCR}(\boldsymbol{X}'')$. 
        Since two adjacent entries in $\mathrm{CCR}(\boldsymbol{X})$ are distinct, we can determine the deleted position $j$ by comparing $\mathrm{CCR}(\boldsymbol{X})$ and $\mathrm{CCR}(\boldsymbol{X}'')_{[1:n-1]}$. 
        We then move the last column of $\boldsymbol{X}''$ to the $j$-th column to obtain the desired array $\boldsymbol{X}'$.
  \item \emph{Recover the array $\boldsymbol{X}$:} By interpreting each row of $\boldsymbol{X}$ as an integer (using a one-to-one mapping), we obtain a vector, denoted as $\mathrm{RIR}(\boldsymbol{X})$, of length $n$ over an alphabet of size $q^n$. 
  We encode $\mathrm{RIR}(\boldsymbol{X})$ into some one-dimensional single deletion correcting code, enabling us to recover $\mathrm{RIR}(\boldsymbol{X})$ from $\mathrm{RIR}(\boldsymbol{X}')$, since $\mathrm{RIR}(\boldsymbol{X}')_{[1:n-1]}$ can be obtained from $\mathrm{RIR}(\boldsymbol{X})$ after one deletion at position $i$.
  Given the bijection between $\mathrm{RIR}(\boldsymbol{X})$ and $\boldsymbol{X}$, knowing $\mathrm{RIR}(\boldsymbol{X})$ allows us to derive the correct $\boldsymbol{X}$.
\end{itemize}

\begin{figure}[t]
\begin{equation*}
\begin{gathered}
  \boldsymbol{X}= 
  \begin{bmatrix}
  \begin{array}{ccc:c:ccc}
    X_{1,1}   & \cdots & X_{1,j-1}   & {\color{red}X_{1,j}}   & X_{1,j+1}   & \cdots & X_{1,n}\\
    \vdots    & \ddots & \vdots      & {\color{red}\vdots}    & \vdots      & \ddots & \vdots\\
    X_{i-1,1} & \cdots & X_{i-1,j-1} & {\color{red}X_{i-1,j}} & X_{i-1,j+1} & \cdots & X_{i-1,n}\\
    \hdashline
    {\color{blue}X_{i,1}}   & {\color{blue}\cdots} & {\color{blue}X_{i,j-1}}     & {\color{green}X_{i,j}}   & {\color{blue}X_{i,j+1}}   & {\color{blue}\cdots} & {\color{blue}X_{i,n}}\\
    \hdashline
    X_{i+1,1} & \cdots & X_{i+1,j-1} & {\color{red}X_{i+1,j}} & X_{i+1,j+1} & \cdots & X_{i+1,n}\\
    \vdots    & \ddots & \vdots      & {\color{red}\vdots}    & \vdots      & \ddots & \vdots\\
    X_{n,1}   & \cdots & X_{n,j-1}   & {\color{red}X_{n,j}}   & X_{n,j+1}   & \cdots & X_{n,n}
  \end{array}
  \end{bmatrix}\\
  \updownarrow \\
  \boldsymbol{X}'= 
  \begin{bmatrix}
  \begin{array}{ccc:c:ccc}
    X_{1,1}   & \cdots & X_{1,j-1}   & {\color{red}X_{1,j}}   & X_{1,j+1}   & \cdots & X_{1,n}\\
    \vdots    & \ddots & \vdots      & {\color{red}\vdots}    & \vdots      & \ddots & \vdots\\
    X_{i-1,1} & \cdots & X_{i-1,j-1} & {\color{red}X_{i-1,j}} & X_{i-1,j+1} & \cdots & X_{i-1,n}\\
    X_{i+1,1} & \cdots & X_{i+1,j-1} & {\color{red}X_{i+1,j}} & X_{i+1,j+1} & \cdots & X_{i+1,n}\\
    \vdots    & \ddots & \vdots      & {\color{red}\vdots}    & \vdots      & \ddots & \vdots\\
    X_{n,1}   & \cdots & X_{n,j-1}   & {\color{red}X_{n,j}}   & X_{n,j+1}   & \cdots & X_{n,n}\\
    \hdashline
    {\color{blue}X_{i,1}}   & {\color{blue}\cdots} & {\color{blue}X_{i,j-1}}     & {\color{green}X_{i,j}}   & {\color{blue}X_{i,j+1}}   & {\color{blue}\cdots} & {\color{blue}X_{i,n}}
  \end{array}
  \end{bmatrix}\\
  \updownarrow \\
  \boldsymbol{X}''= 
  \begin{bmatrix}
  \begin{array}{cccccc:c}
    X_{1,1}   & \cdots & X_{1,j-1}   & X_{1,j+1}   & \cdots & X_{1,n}    & {\color{red}X_{1,j}}\\
    \vdots    & \ddots & \vdots      & \vdots      & \ddots & \vdots     & {\color{red}\vdots}\\
    X_{i-1,1} & \cdots & X_{i-1,j-1} & X_{i-1,j+1} & \cdots & X_{i-1,n}  & {\color{red}X_{i-1,j}} \\
    X_{i+1,1} & \cdots & X_{i+1,j-1} & X_{i+1,j+1} & \cdots & X_{i+1,n}  & {\color{red}X_{i+1,j}} \\
    \vdots    & \ddots & \vdots      & \vdots      & \ddots & \vdots     & {\color{red}\vdots}    \\
    X_{n,1}   & \cdots & X_{n,j-1}   & X_{n,j+1}   & \cdots & X_{n,n}    & {\color{red}X_{n,j}}   \\
    \hdashline
    {\color{blue}X_{i,1}}   & {\color{blue}\cdots} & {\color{blue}X_{i,j-1}}     & {\color{blue}X_{i,j+1}}   & {\color{blue}\cdots} & {\color{blue}X_{i,n}}     & {\color{green}X_{i,j}}   
  \end{array}
  \end{bmatrix}\\
  \updownarrow \\
  \boldsymbol{Y}= 
  \begin{bmatrix}
    X_{1,1}   & \cdots & X_{1,j-1}   & X_{1,j+1}   & \cdots & X_{1,n}   \\
    \vdots    & \ddots & \vdots      & \vdots      & \ddots & \vdots    \\
    X_{i-1,1} & \cdots & X_{i-1,j-1} & X_{i-1,j+1} & \cdots & X_{i-1,n} \\
    X_{i+1,1} & \cdots & X_{i+1,j-1} & X_{i+1,j+1} & \cdots & X_{i+1,n} \\
    \vdots    & \ddots & \vdots      & \vdots      & \ddots & \vdots    \\
    X_{n,1}   & \cdots & X_{n,j-1}   & X_{n,j+1}   & \cdots & X_{n,n}
  \end{bmatrix}
\end{gathered}
\end{equation*} 
\caption{Illustration of the criss-cross deletion process. Assume that $\boldsymbol{Y}$ is obtained form $\boldsymbol{X}$ by deleting its $i$-th row and $j$-th column. We can visualize the process of obtaining $\boldsymbol{Y}$ as follows: First, we move the $i$-th row of $\boldsymbol{X}$ to the bottom, resulting in the array $\boldsymbol{X}'$. Next, we shift the $j$-th column of $\boldsymbol{X}'$ to the rightmost position, yielding $\boldsymbol{X}''$. Finally, $\boldsymbol{Y}$ is the top-left corner subarray of $\boldsymbol{X}''$, with size $(n-1) \times (n-1)$.}
\label{fig:array}
\end{figure}

Before presenting our code construction, we provide a precise definition of $\mathrm{CCR}(\boldsymbol{X})$ and $\mathrm{RIR}(\boldsymbol{X})$.

\begin{definition}  
  Let $\Delta_{n}^{q-1}$ denote the set of all $q$-tuples $(u_0, u_1, \ldots, u_{q-1})$ where each $u_i \geq 0$ and their sum equals $n$.
  For two distinct $\bm{u},\bm{v}\in \Delta_{n}^{q-1}$, we define a \emph{partial order} $\boldsymbol{u}\prec \boldsymbol{v}$ if there exists some $i\in [0:q-1]$ such that $\boldsymbol{u}_j=\boldsymbol{v}_j$ for $j<i$ and $\boldsymbol{u}_i<\boldsymbol{v}_i$.
  The \emph{lexicographic rank} of $\boldsymbol{u}\in \Delta_{n}^{q-1}$, denoted by $\mathrm{Rank}(\bm{u})$, is then defined with respect to this partial order.
  For a given array $\boldsymbol{X} \in \Sigma_q^{n \times n}$, let $\mathrm{CCR}(\boldsymbol{X})$ be its \emph{column composition representation}, which is defined as 
  \begin{align*}
    \mathrm{CCR}(\boldsymbol{X}) := \big(\mathrm{CCR}(X)_1, \mathrm{CCR}(X)_2, \ldots, \mathrm{CCR}(X)_n\big),
  \end{align*}
  where $\mathrm{CCR}(X)_k := \mathrm{Rank}\big( \mathrm{Comp}(\boldsymbol{X}_{[1:n],k})\big)$ represents the rank of the composition of the $k$-th column of $\boldsymbol{X}$ for $k \in [1:n]$.  
\end{definition}

\begin{definition}  
  For a given array $\boldsymbol{X} \in \Sigma_q^{n \times n}$, let $\mathrm{RIR}(\boldsymbol{X})$ be its \emph{row integer representation}, which is defined as
  \begin{align*}
    \mathrm{RIR}(\boldsymbol{X}) := \big(\mathrm{RIR}(X)_1, \mathrm{RIR}(X)_2, \ldots, \mathrm{RIR}(X)_n\big),
  \end{align*}
  where $\mathrm{RIR}(X)_k := \sum_{j=1}^n X_{k,j} q^{n-j}$ represents the integer representation of the $k$-th row of $\boldsymbol{X}$ in base $q$ for $k \in [1:n]$.  
\end{definition}

\begin{definition}\label{def:good}
  An array $\boldsymbol{X} \in \Sigma_q^{n \times n}$ is called \emph{good} if it satisfies the following conditions:
  \begin{itemize}
    \item $\sum_{t=1}^n X_{t,k} \equiv 0 \pmod{q}$ and $\sum_{t=1}^n X_{k,t} \equiv 0 \pmod{q}$ for $k\in [1:n]$;
    \item The compositions of any two adjacent columns of $\boldsymbol{X}$ are distinct, i.e., $\mathrm{CCR}(X)_k \neq \mathrm{CCR}(X)_{k+1}$ for $k \in [1:n-1]$.
  \end{itemize}
\end{definition}

We are now prepared to present our code construction.

\begin{construction}\label{constr1}
  Let $c,d\in \Sigma_n$, we define the code 
  $\mathcal{C}_1(c,d)$ in which each array $\boldsymbol{X}$ satisfies the following constraints:
  \begin{itemize}
    \item $\boldsymbol{X}$ is good;
    \item $\sum_{t=1}^{n-1} t \cdot \alpha\big(\mathrm{CCR}(X)\big)_t \equiv c \pmod{n}$;
    \item $\sum_{t=1}^{n-1} t\cdot \alpha\big(\mathrm{RIR}(X)\big)_t \equiv d \pmod{n}$.
  \end{itemize}
\end{construction}

\begin{theorem}\label{thm:non-binary}
  Let $c,d\in \Sigma_n$, the code $\mathcal{C}_1(c,d)$ defined in Construction \ref{constr1} is a $(1,1)$-cirss-cross deletion correcting code.
\end{theorem}

\begin{proof}
We will demonstrate that $\mathcal{C}_1(c,d)$ is a $(1,1)$-criss-cross deletion correcting code by presenting a decoding algorithm.
Assume that $\boldsymbol{X} \in \mathcal{C}_1(c,d)$ undergoes a row deletion at position $i$ followed by a column deletion at position $j$, resulting in the array $\boldsymbol{Y}$, as illustrated in Figure \ref{fig:array}. 

Observe that for $k\in [1:j-1]\cup [j+1:n]$, the symbol $X_{i,k}$ can be determined as follows:
\begin{equation}\label{eq:row}
  X_{i,k}=- \sum_{t\neq i}X_{t,k} \pmod{q},
\end{equation}
and for $k\in [1:n]$, the symbol $X_{k,j}$ can be computed as:
\begin{equation}\label{eq:column}
  X_{k,j}=- \sum_{t\neq j}X_{k,t} \pmod{q}.
\end{equation}
We define an array $\boldsymbol{X}'' \in \Sigma_q^{n \times n}$ that satisfies the first condition of good arrays, where the subarray indexed by the first $n-1$ rows and the first $n-1$ columns is equal to $\boldsymbol{Y}$.
Clearly, by Figure \ref{fig:array}, $\boldsymbol{X}'' \in \Sigma_q^{n \times n}$ can be obtained from $\boldsymbol{X}$ by moving its $i$-th row to the bottom followed by shifting its $j$-th column to rightmost position.

Next, we consider the column composition representation of $\boldsymbol{X}''$, denoted as $\mathrm{CCR}(\boldsymbol{X}'')$. 
We observe that $\mathrm{CCR}(\boldsymbol{X}'')_{[1:n-1]}$ is obtained from $\mathrm{CCR}(\boldsymbol{X})$ by deleting $\mathrm{CCR}(X'')_n$ at position $j$.
Since $\sum_{t=1}^{n-1} t \cdot \alpha(\mathrm{CCR}(\boldsymbol{X}))_t \equiv c \pmod{n}$, by Corollary \ref{cor:VTdecoder}, we can utilize the non-binary VT decoder to recover $\mathrm{CCR}(\boldsymbol{X})$ from $\mathrm{CCR}(\boldsymbol{X}'')$.
Moreover, given that $\mathrm{CCR}(X)_k \neq \mathrm{CCR}(X)_{k+1}$ for $k \in [1:n-1]$, we can precisely determine the column deletion position $ j $ by comparing $\mathrm{CCR}(\boldsymbol{X})_{[1:n-1]}''$ and the recovered $\mathrm{CCR}(\boldsymbol{X})$.
We then move the last column of $\boldsymbol{X}''$ to the $ j $-th column to obtain the desired array $\boldsymbol{X}'$.
Note that, by Figure \ref{fig:array}, $\boldsymbol{X}'$ can be obtained from $\boldsymbol{X}$ by moving its $i$-th row to the bottom.

Now, we consider the row integer representation of $\boldsymbol{X}'$, denoted as $\mathrm{RIR}(\boldsymbol{X}')$. 
We find that $\mathrm{RIR}(\boldsymbol{X}')_{[1:n-1]}$ is derived from $\mathrm{RIR}(\boldsymbol{X})$ by deleting $\mathrm{RIR}(X')_n$ at position $i$.
Since $\sum_{t=1}^{n-1} t \cdot \alpha(\mathrm{RIR}(\boldsymbol{X}))_t \equiv d \pmod{n}$, by Corollary \ref{cor:VTdecoder}, we can again employ the non-binary VT decoder to recover $\mathrm{RIR}(\boldsymbol{X})$ from $\mathrm{RIR}(\boldsymbol{X}')$.
This process yields the correct $\boldsymbol{X}$, as the function $\mathrm{RIR}(\cdot)$ is a bijection.
The proof is completed.
\end{proof}

\begin{remark}
The decoding algorithm described in the proof of Theorem \ref{thm:non-binary} operates with a time complexity of $O(n^2)$, as detailed below:
\begin{itemize}
    \item Each $X_{i,k}$ and $X_{k,j}$ can be computed in $O(n)$ time. Since a total of $2n-1$ values need to be calculated, the overall complexity for this step is $O(n^2)$.
    \item The composition of each column can also be computed in $O(n)$ time. With $n$ columns to calculate, the overall complexity for this step is $O(n^2)$.
    \item A comparison between the lexicographic ranks of two compositions can be performed in $O(1)$ time. With $O(n)$ pairs requiring comparison, the total complexity for this step is $O(n)$.
    \item The non-binary VT decoder operates in $O(n)$ time.
    \item The location of the column deletion can be determined in $O(n)$ time when obtaining $\mathrm{CCR}(\boldsymbol{X})_{[1:n-1]}''$ and $\mathrm{CCR}(\boldsymbol{X})$.
    \item The integer representation of each row can be computed in $O(n)$ time. Since there are $n$ rows to calculate, the overall complexity for this step is $O(n^2)$.
\end{itemize}
In summary, the total complexity is linear with respect to the codeword size, which makes it optimal for two-dimensional scenarios.
\end{remark}

Prior to the calculation of the number of good arrays, we present some essential definitions.
Recall that $\Delta_{n}^{q-1}$ denotes the set of all $q$-tuples $(u_0, u_1, \ldots, u_{q-1})$ where each $u_i \geq 0$ and their sum equals $n$.
\begin{definition}
  For any two sequences $\bm{u}, \bm{v} \in \Delta_{q}^{n-2}$, their $L_1$ distance is defined as $d_{L_1}(\bm{u}, \bm{v}) = \frac{1}{2}\sum_{i=0}^{q-1} |u_i - v_i|$.
\end{definition}

\begin{remark}\label{rmk:L_1}
  For any two sequences $\bm{u}, \bm{v} \in \Delta_{q}^{n-2}$, if $d_{L_1}(\bm{u}, \bm{v}) \geq 3$, then for any $\bm{u}', \bm{v}' \in \Delta_{n}^{q-1}$ such that $u_i' \geq u_i$ and $v_i' \geq v_i$ for $i \in [0:q-1]$, it follows that $\bm{u}' \neq \bm{v}'$.
Furthermore, the cardinality of the set of sequences in $\Delta_{q}^{n-2}$ whose $L_1$ distance from $\bm{u}$ is less than three is bounded above by $5$ if $q=2$ and $\binom{q+1}{2}^2\leq \frac{q^4-1}{2}$ if $q\geq 3$.
\end{remark}

\begin{definition}\label{def:f}
  Let $\bm{Z} \in \Sigma_q^{(n-2)\times n}$. We define a function $f(\cdot)$ that maps $\bm{Z}$ to an array $f(\bm{Z}) \in \{\Sigma_q\cup \{\ast\}\}^{n \times n}$. The construction proceeds for each column $k \in [1:n]$ as follows:
  \begin{enumerate}
    \item The entries $f(\bm{Z})_{k,k}$ and $f(\bm{Z})_{k,(k+1 \mod{n})}$ are set to `$\ast$'.
    \item The remaining entries of the $k$-th column of $f(\bm{Z})$ are populated by the corresponding entries of the $k$-th column of $\bm{Z}$, such that $f(\bm{Z})_{[1:n]\setminus\{k, (k+1 \mod{n})\},k}=\bm{Z}_{[1:n-2],k}$.
  \end{enumerate}
\end{definition}

\begin{definition}\label{def:g}
  For any $\bm{Z} \in \Sigma_q^{(n-2)\times n}$ and $\sigma \in \Sigma_q$, we define the array $g(\bm{Z}, \sigma) \in \Sigma_q^{n \times n}$ through the following sequential construction:
\begin{enumerate}
    \item Initialize an intermediate array $\bm{U}$ by setting $\bm{U} = f(\bm{Z})$.
    \item Update the entry $U_{1,n}$ to $\sigma$.
    \item Iterate for each column index $k$ from $1$ to $n-1$:
    \begin{itemize}
        \item Adjust $U_{k,k}$ such that the sum of the entries in the $k$-th row of $\bm{U}$ is a multiple of $q$. Specifically, set $U_{k,k} \equiv -\sum_{j\neq k} U_{k,j} \pmod{q}$.
        \item Adjust $U_{k+1,k}$ such that the sum of the entries in the $k$-th column of $\bm{U}$ is a multiple of $q$. Specifically, set $U_{k+1,k} \equiv -\sum_{j\neq k+1} U_{j,k} \pmod{q}$.
    \end{itemize}
    \item Finally, adjust $U_{n,n}$ such that the sum of the entries in the last row of $\bm{U}$ is a multiple of $q$. That is, define $U_{n,n} \equiv -\sum_{j\neq n} U_{n,j} \pmod{q}$. Note that, with this final adjustment, the sum of the entries in the last column of $\bm{U}$ also becomes a multiple of $q$.
    \item The resulting array $\bm{U}$ is defined as $g(\bm{Z},\sigma)$, i.e., $g(\bm{Z},\sigma)=\bm{U}$.
\end{enumerate}
\end{definition}

\begin{lemma}\label{rmk:good}
  For any $\bm{Z} \in \Sigma_q^{(n-2)\times n}$ and $\sigma \in \Sigma_q$, if $d_{L_1}\big(\mathrm{Comp}(\bm{Z}_{[1:n-2],k}), \mathrm{Comp}(\bm{Z}_{[1:n-2],k+1})\big)\geq 3$ for $k\in [1:n-1]$, then $g(\bm{Z}, \sigma)$ is a good array. 
\end{lemma}

\begin{IEEEproof}
    Let $\bm{U}=g(\bm{Z}, \sigma)$.
    Since $d_{L_1}\big(\mathrm{Comp}(\bm{Z}_{[1:n-2],k}), \mathrm{Comp}(\bm{Z}_{[1:n-2],k+1})\big)\geq 3$, by Remark \ref{rmk:L_1}, we have $\mathrm{Comp}(\bm{U}_{[1:n],k})\neq \mathrm{Comp}(\bm{U}_{[1:n],k+1})$. 
    Moreover, an inspection of Steps 3) and 4) of Definition \ref{def:g} readily confirms that $\sum_{t=1}^n U_{t,k} \equiv 0 \pmod{q}$ and $\sum_{t=1}^n U_{k,t} \equiv 0 \pmod{q}$ for each $k\in [1:n]$. Consequently, $\bm{U}$ constitutes a good array.
\end{IEEEproof}

Now, we are ready to calculate the number of good arrays.
In the following, we will assume $n-2$ is a multiple of $q$. In cases where this assumption does not hold, we may need to use either the floor or ceiling function in specific instances.
While this could introduce more intricate notation, it does not affect the asymptotic analysis when using Stirling’s formula.

\begin{lemma}\label{lem:good}
  When $q\geq 3$, $n\geq \frac{7}{2}q^4+1$, and $q|(n-2)$, the number of $q$-ary good arrays of size $n\times n$ is at least $3^{-\frac{q^4}{2}}\cdot q^{n^2-2n+1}$.
\end{lemma}

\begin{proof}
Let $\mathcal{S}$ be the set of arrays in $\Sigma_q^{(n-2)\times n}$ such that the $L_1$ distance between the compositions of any two adjacent columns is at least three.
By Lemma \ref{rmk:good}, $g(\bm{Z}, \sigma)$ is a good array for $\bm{Z}\in \mathcal{S}$ and $\sigma\in \Sigma_q$.
Since $g(\cdot,\cdot)$ is an injection, we can conclude that the number of good arrays is at least $q\cdot |\mathcal{S}|$.
It remains to show that $|\mathcal{S}|\geq 3^{-\frac{q^4}{2}}\cdot q^{n^2-2n}$.

Let $M$ denote the maximum size of a set of $q$-ary sequences of length $n-2$ with identical composition. 
By Remark \ref{rmk:L_1}, we can compute
\begin{equation}\label{eq:number}
\begin{aligned}
  |\mathcal{S}|
  \geq q^{n-2}\left(q^{n-2}-\frac{q^4-1}{2} M \right)^{n-1}
  = q^{n^2-2n} \left( 1-\frac{(q^4-1)M}{2q^{n-2}} \right)^{n-1}.
\end{aligned}
\end{equation}
In what follows, we will demonstrate that $M\leq \frac{q^{n-2}}{n-2}\leq \frac{q^{n+2}}{(q^4-1)(n-1)}$.
Once this is established, we can further evaluate Equation (\ref{eq:number}) as follows:
\begin{align*}
    |\mathcal{S}|
    \geq q^{n^2-2n} \left( 1-\frac{q^4}{2(n-1)} \right)^{n-1}
    \geq 3^{-\frac{q^4}{2}}\cdot q^{n^2-2n},
\end{align*}
where the last inequality follows from the fact that $\left(1-\frac{q^4}{2(n-1)} \right)^{2(n-1)/q^4}$ is an increasing function of $n$ that is always great than or equal to $\left(1-\frac{1}{7}\right)^7\geq \frac{1}{3}$ when $n\geq \frac{7}{2}q^4+1$.

Next, we show that the inequality $M\leq \frac{q^{n-2}}{n-2}$ holds. By definition, we have
\begin{align*}
  M\leq \binom{n-2}{\frac{n-2}{q}, \frac{n-2}{q}, \ldots, \frac{n-2}{q}}= \frac{(n-2)!}{\left(\frac{n-2}{q}!\right)^q}.
\end{align*}
By \cite[Page 6, Equation (1.7)]{Jukna}, we can use the well-known Stirling's formula to obtain the following bounds for $n!$:
\begin{equation}\label{eq:stirling}
    \sqrt{2\pi n} \left(\frac{n}{e}\right)^n e^{\frac{1}{12n+1}}< n! \leq \sqrt{2\pi n}\left(\frac{n}{e} \right)^n e^{\frac{1}{12n}}.
\end{equation}
We can then compute
\begin{align}\label{eq:M}
  M\leq \frac{(n-2)!}{\left(\frac{n-2}{q}!\right)^q} 
  &\leq \frac{\sqrt{2\pi (n-2)}\left(\frac{n-2}{e} \right)^{n-2} e^{\frac{1}{12(n-2)}}}{\left( \sqrt{2\pi \frac{n-2}{q}} \left(\frac{n-2}{eq}\right)^\frac{n-2}{q} e^{\frac{1}{12\frac{n-2}{q}+1}}\right)^q} \nonumber \\
  &= \frac{\sqrt{2\pi (n-2)}}{\sqrt{\left(2\pi \frac{n-2}{q}\right)^q}} \cdot \frac{\left(\frac{n-2}{e} \right)^{n-2}}{\left(\frac{n-2}{eq}\right)^{n-2}} \cdot \frac{e^{\frac{1}{12(n-2)}}}{e^{\frac{q^2}{12(n-2)+q}}} \nonumber\\
  &\leq \sqrt{\frac{q^q}{(2\pi)^{q-1}(n-2)^{q-3}}} \cdot \frac{q^{n-2}}{n-2} \\
  &\leq \frac{q^{n-2}}{n-2}\nonumber,
\end{align}
where the last inequality follows by the fact that 
\begin{itemize}
  \item when $q=3$, $\sqrt{\frac{q^q}{(2\pi)^{q-1}(n-2)^{q-3}}}= \frac{3\sqrt{3}}{2\pi}\leq 1$;
  \item when $q\geq 4$ and $n\geq \frac{7}{2}q^4+1\geq q^4+2$, $\sqrt{\frac{q^q}{(2\pi)^{q-1}(n-2)^{q-3}}}\leq \sqrt{\frac{q^q}{(2\pi)^{q-1}q^{4(q-3)}}}\leq \sqrt{\frac{1}{(2\pi)^{q-1}}}\leq 1$.
\end{itemize}
The proof is completed.
\end{proof}

\begin{remark}\label{rmk:binary}
  When $q=2$, we can establish that $M\leq K\cdot \frac{q^n}{\sqrt{n}}$ for some constant $K$, by utilizing Equation (\ref{eq:M}).  In this scenario, there exists some constant $K'$ such that the quantity $q^{n-2}\left(q^{n-2}-\frac{q^4-1}{2}M \right)^{n-1}$ can be estimated as follows:
  \begin{align*}
    q^{n-2}\left(q^{n-2}-\frac{q^4-1}{2} M \right)^{n-1}\sim q^{n^2-2n} \left(1-\frac{K'}{\sqrt{n}} \right)^{n}\stackrel{(\star)}{\sim} \frac{q^{n^2-2n}}{e^{\sqrt{n}K'}},
  \end{align*}
  where $(\star)$ follows from the fact that $\left(1-\frac{1}{m} \right)^{m}$ approaches $\frac{1}{e}$ as $m$ increases.
  The resulting redundancy is $2n+O(\sqrt{n})$.
  Consequently, the code defined in Construction \ref{constr1} may be significantly distant from the optimal redundancy.
\end{remark}

With the help of Lemma \ref{lem:good}, we are now prepared to calculate the redundancy of $\mathcal{C}_1(c,d)$.

\begin{corollary}
  When $q\geq 3$, $n\geq \frac{7}{2}q^4+1$, and $q|(n-2)$, there exists a choice of parameters such that the redundancy of $\mathcal{C}_1(c,d)$ (as defined in Construction \ref{constr1}) is at most $(2n-1)\log q+2\log n+\frac{q^4}{2}\log 3$, which is optimal up to an additive constant.
\end{corollary}

\begin{proof}
By considering all possible choices for $c, d \in \Sigma_n$, we find that the total number of codes $\mathcal{C}_1(c, d)$ is $n^2$, which forms a partition of good arrays in $\Sigma_q^{n \times n}$. Consequently, by the pigeonhole principle, there exists at least one choice of $c, d \in \Sigma_n$ such that the code $\mathcal{C}_1(c, d)$ contains at least $\frac{3^{-\frac{q^4}{2}}\cdot q^{n^2-2n+1}}{n^2}$ codewords. Thus, the conclusion follows.
\end{proof}

\section{Construction of Optimal \texorpdfstring{$(1,1)$}{}-Criss-Cross deletion correcting  Codes over Arbitrary Alphaets}\label{sec:binary}%

This section presents a construction of $(1,1)$-criss-cross deletion correcting codes that exhibits optimal redundancy (up to an additive constant term) over arbitrary alphabets. As noted in Remark \ref{rmk:binary}, the code defined in Construction \ref{constr1} may be inefficient for the binary alphabet due to the high redundancy associated with good arrays. Consequently, we need to introduce several new ideas.  

Our main strategy is to identify appropriate constraints for each codeword over the binary alphabet that require a redundancy of $O(1)$ and fulfill the same role as the good array introduced in Definition \ref{def:good}. By relaxing the requirements of the good array and adding several additional constraints, we achieve this goal and derive an optimal $(1,1)$-criss-cross deletion correcting code. 
Let $\boldsymbol{X}, \boldsymbol{X}', \boldsymbol{X}'', \boldsymbol{Y}$ be defined in Figure \ref{fig:array}, the detailed approach is explained  as follows:

\begin{itemize}
  \item \emph{Determine the missing symbols in each row and column:} 
  This step follows the procedure introduced in the previous section, allowing us to derive $\boldsymbol{X}''$ from $\boldsymbol{Y}$.
  
  \item \emph{Locate the row deletion position $i$ and the column deletion position $j$ within small intervals:} Using $\boldsymbol{X}''$, we can determine the composition of each column of $\boldsymbol{X}$. 
  This enables us to compute $\mathrm{CCR}(\boldsymbol{X})$. We require that $\mathrm{CCR}(\boldsymbol{X})$ belongs to a one-dimensional single deletion correcting code, with the additional constraint that no $P+1$ consecutive entries are all identical.
  Following the approach of the previous section, we can recover $\mathrm{CCR}(\boldsymbol{X})$. However,  since adjacent entries in $\mathrm{CCR}(\boldsymbol{X})$ may be identical, the column deletion position $j$ can only be localized to an interval of length $P$ in the worst case. 
  Similarly, by applying the same constraints to the rows of $\boldsymbol{X}$ as those applied to the columns, we can also locate the row deletion position $i$ within an interval of length $P$ in the worst case. 
  
  \item \emph{Determine the column deletion position $j$:} Let $\ell \geq P-1$ be an integer satisfying $\frac{n}{\ell} \geq 3$.
      We consider three subarrays: $\boldsymbol{X}_{[1:\ell],[1:n]}$, $\boldsymbol{X}_{[\ell+1,2\ell],[1:n]}$, and $\boldsymbol{X}_{[2\ell+1:3\ell],[1:n]}$. Since $P$ adjacent rows can span at most two of these subarrays, we independently encode these subarrays to ensure that at least one subarray, without loss of generality, assumed to be the first subarray $\boldsymbol{X}_{[1:\ell],[1:n]}$, does not experience a row deletion.
      We observe that $\boldsymbol{X}_{[1:\ell],[1:n-1]}''$ is derived from $\boldsymbol{X}_{[1:\ell],[1:n]}$ by deleting its $j$-th column $\boldsymbol{X}_{[1:\ell],n}''$. 
      By mapping each column of $\boldsymbol{X}_{[1:\ell],[1:n]}$ as an integer (using a bijective mapping), we obtain a vector 
      \[\mathrm{CIR}(\boldsymbol{X}_{[1:\ell],[1:n]}):= \mathrm{RIR}\big((\boldsymbol{X}_{[1:\ell],[1:n]})^T\big),\] 
      which has length $ n $ and is defined over an alphabet of size $ q^{\ell} $.
      We then require that $\mathrm{CIR}(\boldsymbol{X}_{[1:\ell],[1:n]})$ belongs to some one-dimensional $P$-bounded single deletion correcting code.
      This constraint allows us to recover $\mathrm{CIR}(\boldsymbol{X}_{[1:\ell],[1:n]})$ given $\mathrm{CIR}(\boldsymbol{X}_{[1:\ell],[1:n]}'')$ and an interval of length at most $P$ containing the column deletion position $ j $.
       Additionally, we require that any two consecutive entries in $\mathrm{CIR}(\boldsymbol{X}_{[1:\ell],[1:n]})$ are distinct. This condition enables us to precisely determine the column deletion position $ j $ by comparing $\mathrm{CIR}(\boldsymbol{X}_{[1:\ell],[1:n]})$ with $\mathrm{CIR}(\boldsymbol{X}_{[1:\ell],[1:n]}'')$.
       Finally, we move the last column of $\boldsymbol{X}''$ to the $j$-th column to obtain the desired array $\boldsymbol{X}'$.

  \item \emph{Recover the array $\boldsymbol{X}$:}
      From the recovered $\boldsymbol{X}'$, we can compute its row integer representation $\mathrm{RIR}(\boldsymbol{X}')$.
      Observe that $\mathrm{RIR}(\boldsymbol{X}')_{[1:n-1]}$ is obtained from $\mathrm{RIR}(\boldsymbol{X})$ by deleting its $i$-th entry $\mathrm{RIR}(X')_{n}$.
      By further requiring that $\mathrm{RIR}(\boldsymbol{X})$ belongs to some one-dimensional $P$-bounded single deletion correcting code, we can recover $\mathrm{RIR}(\boldsymbol{X})$ given $\mathrm{RIR}(\boldsymbol{X}')$ and an interval of length at most $P$ containing the row deletion position $ i $.
      This process yields the correct $\boldsymbol{X}$, as the function $\mathrm{RIR}(\cdot)$ is a bijection.
\end{itemize}

Before presenting our code construction, we provide the following definition to be used.

\begin{definition}\label{def:valid}
Let $P, \ell$ be positive integers satisfying $n\geq 3\ell$ and $\ell\geq P-1$, an array $\boldsymbol{X} \in \Sigma_q^{n \times n}$ is termed \emph{$(P,\ell)$-valid} if it satisfies the following conditions:
\begin{itemize}
    \item $\sum_{t=1}^n X_{t,k} \equiv 0 \pmod{q}$ and $\sum_{t=1}^n X_{k,t} \equiv 0 \pmod{q}$ for $k\in [1:n]$;
    \item The compositions of any $P+1$ consecutive columns of $\boldsymbol{X}$ are not all identical;
    \item The compositions of any $P+1$ consecutive rows of $\boldsymbol{X}$ are not all identical;
    \item For $k\in [1:3]$, in rows $(k-1)\ell+1$ to $k\ell$ of $ \boldsymbol{X} $, any two adjacent columns are distinct.
\end{itemize}
Moreover, if $\boldsymbol{X}$ satisfies only the first and last conditions, we say that $\boldsymbol{X}$ is \emph{$\ell$-weakly-valid}.
\end{definition}

Our code construction is presented as follows.

\begin{construction}\label{constr2}
  Let $P, \ell$ be positive integers satisfying $n\geq 3\ell$ and $\ell\geq P-1$. Let $\boldsymbol{c}= (c_1,c_2)\in \Sigma_n^2$, $\boldsymbol{d}=(d_1,d_2,d_3,d_4)\in \Sigma_2^4$, and $\boldsymbol{d}'=(d_1',d_2',d_3',d_4')\in \Sigma_{P+1}^4$, we define the code $\mathcal{C}_2(\boldsymbol{c},\boldsymbol{d},\boldsymbol{d}')$ in which each array $\boldsymbol{X}$ satisfies the following constraints:
  \begin{itemize}
    \item $\boldsymbol{X}$ is $(P,\ell)$-valid;
    \item $\sum_{t=1}^{n-1} t \cdot \alpha\big(\mathrm{CCR}(X)\big)_t \equiv c_1 \pmod{n}$ and $\sum_{t=1}^{n-1} t \cdot \alpha\big(\mathrm{CCR}(X^T)\big)_t \equiv c_2 \pmod{n}$;
    \item $\sum_{t=1}^{n-1}  \alpha\big(\mathrm{CIR}(\boldsymbol{X}_{[(k-1)\ell+1:k\ell], [1:n]})\big)_t \equiv d_k \pmod{2}$ and $\sum_{t=1}^{n-1} t\cdot \alpha\big(\mathrm{CIR}(\boldsymbol{X}_{[(k-1)\ell+1:k\ell], [1:n]})\big)_t \equiv d_k' \pmod{P+1}$ for $k\in [1:3]$;
    \item $\sum_{t=1}^{n-1}  \alpha\big(\mathrm{RIR}(\boldsymbol{X})\big)_t \equiv d_4 \pmod{2}$ and $\sum_{t=1}^{n-1} t\cdot \alpha\big(\mathrm{RIR}(\boldsymbol{X})\big)_t \equiv d_4' \pmod{P+1}$.
  \end{itemize}
\end{construction}

\begin{theorem}\label{thm:binary}
  Let $P, \ell$ be positive integers satisfying $n\geq 3\ell$ and $\ell\geq P-1$. Let $\boldsymbol{c}= (c_1,c_2)\in \Sigma_n^2$, $\boldsymbol{d}=(d_1,d_2,d_3,d_4)\in \Sigma_2^4$, and $\boldsymbol{d}'=(d_1',d_2',d_3',d_4')\in \Sigma_{P+1}^4$, the code $\mathcal{C}_2(\boldsymbol{c},\boldsymbol{d},\boldsymbol{d}')$ defined in Construction \ref{constr2} is a $(1,1)$-cirss-cross deletion correcting code.
\end{theorem}

\begin{proof}
We will demonstrate that $\mathcal{C}_2(\boldsymbol{c},\boldsymbol{d},\boldsymbol{d}')$ is a $(1,1)$-criss-cross deletion correcting code by presenting a decoding algorithm.
Assume that $\boldsymbol{X} \in \mathcal{C}_2(\boldsymbol{c},\boldsymbol{d},\boldsymbol{d}')$ undergoes a row deletion at position $i$ followed by a column deletion at position $j$, resulting in the array $\boldsymbol{Y}$, as illustrated in Figure \ref{fig:array}. 

We define an array $\boldsymbol{X}'' \in \Sigma_q^{n \times n}$ that satisfies the first condition of valid arrays, where the subarray indexed by the first $n-1$ rows and the first $n-1$ columns is equal to $\boldsymbol{Y}$.
By Equations (\ref{eq:row}) and (\ref{eq:column}), we can conclude that $\boldsymbol{X}'' \in \Sigma_q^{n \times n}$ is obtained from $\boldsymbol{X}$ by moving its $i$-th row to the bottom followed by shifting its $j$-th column to rightmost position.

Next, we consider the column composition representation of $\boldsymbol{X}''$, denoted as $\mathrm{CCR}(\boldsymbol{X}'')$. 
We observe that $\mathrm{CCR}(\boldsymbol{X}'')_{[1:n-1]}$ is obtained from $\mathrm{CCR}(\boldsymbol{X})$ by deleting $\mathrm{CCR}(X'')_n$ at position $j$.
Since $\sum_{t=1}^{n-1} t \cdot \alpha(\mathrm{CCR}(\boldsymbol{X}))_t \equiv c_1 \pmod{n}$, by Corollary \ref{cor:VTdecoder}, we can utilize the non-binary VT decoder to recover $\mathrm{CCR}(\boldsymbol{X})$ from $\mathrm{CCR}(\boldsymbol{X}'')$.
Furthermore, since $\boldsymbol{X}$ is valid, the compositions of any $P+1$ consecutive columns of $\boldsymbol{X}$ are not all identical.
As a result, in the worst case, we can locate the column deletion position $j$ in an interval of length $P$ by comparing $\mathrm{CCR}(\boldsymbol{X})_{[1:n-1]}''$ with the recovered $\mathrm{CCR}(\boldsymbol{X})$.
Similarly, we can locate the row deletion position $i$ within an interval of length $P$ in the worst case.

Since we have identified $P\leq \ell-1$ consecutive positions, one of which is $ i $, there exists some $ k \in [1:3] $ such that $ i \not\in [(k-1)\ell+ 1: k\ell] $. Without loss of generality assume $k=1$, then the subarray $\boldsymbol{X}_{[1:\ell],[1:n]}''$ does not experience a row deletion. Consequently, the subarray $\boldsymbol{X}_{[1: \ell],[1:n-1]}''$ is obtained from $\boldsymbol{X}_{[1: \ell],[1:n]}$ by deleting its $ j $-th column, $\boldsymbol{X}_{[1: \ell],n}''$.
Now, consider the vector $\mathrm{CIR}(\boldsymbol{X}_{[1: \ell],[1:n]})$, which records the integer representation of each column of $\boldsymbol{X}_{[1: \ell],[1:n]}$. It follows that $\mathrm{CIR}(\boldsymbol{X}_{[1: \ell],[1:n-1]}'')$ is obtained from $\mathrm{CIR}(\boldsymbol{X}_{[1: \ell],[1:n]})$ by deleting its $ j $-th entry, $\mathrm{CIR}(\boldsymbol{X}_{[1: \ell],n}'')$.
Recall that we have identified an interval of length $P$ containing $ j $. 
By leveraging the conditions $\sum_{t=1}^{n-1}  \alpha\big(\mathrm{CIR}(\boldsymbol{X}_{[1:\ell], [1:n]})\big)_t \equiv d_1 \pmod{2}$ and $\sum_{t=1}^{n-1} t\cdot \alpha\big(\mathrm{CIR}(\boldsymbol{X}_{[1:\ell], [1:n]})\big)_t \equiv d_1' \pmod{P+1}$, we can utilize the non-binary $P$-bounded VT decoder to recover $\mathrm{CIR}(\boldsymbol{X}_{[1: \ell],[1:n]})$ from $\mathrm{CIR}(\boldsymbol{X}_{[1: \ell],[1:n]}'')$.
Since $\boldsymbol{X}$ is valid (in the first $\ell$ rows of $ \boldsymbol{X} $, any two adjacent columns are distinct), we can precisely determine the column deletion position $ j $.
We then move the last column of $\boldsymbol{X}''$ to the $ j $-th column to obtain the desired array $\boldsymbol{X}'$.

Now, we consider the row integer representation of $\boldsymbol{X}'$, denoted as $\mathrm{RIR}(\boldsymbol{X}')$.
Clearly, $\mathrm{RIR}(\boldsymbol{X}')_{[1:n-1]}$ is obtained from $\mathrm{RIR}(\boldsymbol{X})$ by deleting its $i$-th entry $\mathrm{RIR}(X')_n$.
Recall that we have identified an interval of length $P$ containing $ i $.  
Similar to the previous discussion, we can recover $\mathrm{RIR}(\boldsymbol{X})$ using the conditions $\sum_{t=1}^{n-1}  \alpha\big(\mathrm{RIR}(\boldsymbol{X})\big)_t \equiv d_4 \pmod{2}$ and $\sum_{t=1}^{n-1} t\cdot \alpha\big(\mathrm{RIR}(\boldsymbol{X})\big)_t \equiv d_4' \pmod{P+1}$.
This process yields the correct $\boldsymbol{X}$, as the function $\mathrm{RIR}(\cdot)$ is a bijection.
The proof is completed.
\end{proof}

\begin{remark}\label{rmk:inversion}
When $P=2$, we can also use the code presented in Corollary \ref{cor:inversion} to correct a single deletion given the deleted symbol and an interval of length $P$ containing the erroneous coordinate.
\end{remark}

\begin{remark}\label{rmk:row}
  If we further require that any two consecutive rows in the array $\boldsymbol{X}\in \mathcal{C}_2(\boldsymbol{c},\boldsymbol{d},\boldsymbol{d}')$ are distinct, then decoding $\boldsymbol{X}$, which is distorted by a $(1,1)$-criss-cross deletion, allows us to precisely identify the row and column deletion positions.
  This property will be used later in the construction of the criss-cross burst-deletion correcting code.  
\end{remark}

\begin{remark}
The decoding algorithm described in the proof of Theorem \ref{thm:binary} operates with a time complexity of $O(n^2)$, as detailed below:
\begin{itemize}
    \item Each $X_{i,t}$ and $X_{t,j}$ can be computed in $O(n)$ time. Since a total of $2n-1$ values need to be calculated, the overall complexity for this step is $O(n^2)$.
    \item The composition of each row and column can also be computed in $O(n)$ time. With $n$ rows and $n$ columns to calculate, the overall complexity for this step is $O(n^2)$.
    \item Comparing the lexicographic ranks of two compositions can be done in $O(1)$ time. With $O(n)$ pairs requiring comparison, the overall complexity for this step is $O(n)$.
    \item The non-binary VT decoder operates in $O(n)$ time.
    \item Locating the column deletion position $ j $ within an interval of length at most $P$ can be done in $ O(n) $ time when obtaining $\mathrm{CCR}(\boldsymbol{X})_{[1:n-1]}''$ and $\mathrm{CCR}(\boldsymbol{X})$. Similarly, locating the row deletion position $ i $ within an interval of length at most $P$ can also be accomplished in $ O(n) $ time. 
    \item The integer representation of each column of $\boldsymbol{X}_{[(k-1)\ell + 1: k\ell],[1:n]}''$ can be computed in $O(n)$ time. Since there are $n$ columns to calculate, the overall complexity for this step is $O(n^2)$. Similarly, the integer representations of all rows of $\boldsymbol{X}'$ can be computed in $O(n^2)$ time.
    \item The non-binary $P$-bounded VT decoder operates in $O(n)$ time.
\end{itemize}
In summary, the total complexity is linear with respect to the codeword size, which makes it optimal for two-dimensional scenarios.
\end{remark}


Similar to Lemma \ref{rmk:good}, we identify conditions for an array to be $(P,\ell)$-valid.

\begin{lemma}\label{rmk:valid}
  Let $P, \ell$ be positive integers satisfying $n\geq 3\ell$ and $\ell\geq P-1$.
  For any $\bm{Z} \in \Sigma_q^{(n-2)\times n}$ and $\sigma \in \Sigma_q$, let $\bm{V}=f(\bm{Z})$, where $f(\cdot)$ is defined in Definition \ref{def:f}. 
  If $\bm{V}$ satisfies the following conditions:
  \begin{enumerate}
    \item for $t\in [1:n-P]$, there exists some $t',t''\in [0:P]$ such that $d_{L_1}\big(\mathrm{Comp}(\bm{V}_{[1:n],t}), \mathrm{Comp}(\bm{V}_{[1:n],t+t'})\big)\geq 3$ and $d_{L_1}\big(\mathrm{Comp}(\bm{V}_{t,[1:n]}), \mathrm{Comp}(\bm{V}_{t+t'',[1:n]})\big)\geq 3$ (the symbol `$\ast$' is disregarded when computing the composition);
    \item for $t\in [1:n-1]$ and $k\in [1:3]$, there exists some $t_k\in [(k-1)\ell+1:k\ell]\setminus \{t,t+1,(t+2 \mod{n})\}$ such that $V_{t_k,t}\neq V_{t_k,t+1}$;
  \end{enumerate}
  then $g(\bm{Z}, \sigma)$ is $(P,\ell)$-valid, where $g(\cdot,\cdot)$ is defined in Definition \ref{def:g}. 
\end{lemma}

\begin{IEEEproof}
    Let $\bm{U}=g(\bm{Z}, \sigma)$. Firstly, an inspection of Steps 3) and 4) of Definition \ref{def:g} readily confirms that $\sum_{t=1}^n U_{t,k} \equiv 0 \pmod{q}$ and $\sum_{t=1}^n U_{k,t} \equiv 0 \pmod{q}$ for each $k\in [1:n]$.
    
     Secondly, for $t\in [1:n-P]$, since $d_{L_1}\big(\mathrm{Comp}(\bm{V}_{[1:n],t}), \mathrm{Comp}(\bm{V}_{[1:n],t+t'})\big)\geq 3$, by Remark \ref{rmk:L_1}, we have $\mathrm{Comp}(\bm{U}_{[1:n],t})\neq \mathrm{Comp}(\bm{U}_{[1:n],t+t'})$. 
     This implies that the compositions of any $P+1$ consecutive columns of $\boldsymbol{U}$ are not all identical.
     Similarly, we can conclude that the compositions of any $P+1$ consecutive rows of $\boldsymbol{U}$ are also not all identical.
     
     Lastly, for $t\in [1:n-1]$ and $k\in [1:3]$, since $V_{t_k,t}\neq V_{t_k,t+1}$, where $t_k\notin \{t,t+1,(t+2 \mod{n})\}$, i.e., $V_{t_k,t}\neq \ast$ and $V_{t_k,t+1}\neq \ast$, we have $U_{t_k,t}\neq U_{t_k,t+1}$.
     This implies that $\bm{U}_{[(k-1)\ell+1:k\ell],t}\neq \bm{U}_{[(k-1)\ell+1:k\ell],t+1}$.
     Consequently, $\bm{U}$ is $(P,\ell)$-valid.
\end{IEEEproof}

We are now ready to calculate the number of valid arrays.
We assume that $n-2$ is a multiple of $q$. In cases where this assumption does not hold, we may need to use either the floor or ceiling function in specific instances.
While this may introduce more intricate notation, it does not influence the asymptotic analysis when employing Stirling’s formula.

\begin{lemma}\label{lem:valid}
  Assume $n\geq 3\ell$, $q|(n-2)$, and $\ell\geq P-1\geq 1$.
  Let $N$ denote the number of $\ell$-valid arrays of size $n\times n$, the following statements hold:
  \begin{itemize}
      \item if $q=2$, then $N\geq \big(1-\frac{3n}{2^{\ell-3}}- \frac{2\cdot(50/\pi)^{P/2}}{(n-2)^{(P-2)/2}} \big)\cdot 2^{n^2-2n+1}$ and $N\geq \frac{1}{5}\cdot 2^{n^2-2n+1}$ when $\ell\geq \log n+9$, $P\geq 8$, and $n\geq 58$;
      \item if $q= 3$, then $N\geq \big(1-\frac{3n}{3^{\ell-3}}- \frac{2\cdot 10800^{P/2}}{\pi^P(n-2)^{P-1}}\big)\cdot 3^{n^2-2n+1}$ and $N \geq \frac{1}{5} \cdot 3^{n^2-2n+1}$ when $\ell\geq \log n+5$, $P\geq 8$, and $n\geq 70$;
      \item if $q\geq 4$, then $N\geq \big(1-\frac{3n}{q^{\ell-3}}- \frac{2\cdot q^{(q+8)P/2}}{(n-2)^{(q-1)P/2-1}}\big)\cdot q^{n^2-2n+1}$ and $N\geq \frac{1}{5} \cdot q^{n^2-2n+1}$ when $\ell\geq \log n+4$, $P\geq 8$, and $n\geq q^5+2$.
  \end{itemize}
\end{lemma}

\begin{proof}
    We prove the lemma by a probabilistic analysis.
    Choose an array $\boldsymbol{Z}\in \Sigma_q^{(n-2)\times n}$ uniformly at random, let $\mathbb{P}$ be the probability that $\bm{V}=f(\bm{Z})$ does not satisfy the conditions outlined in Lemma \ref{rmk:valid}.
    Since $g(\cdot,\cdot)$ is an injection, by Lemma \ref{rmk:valid}, we can calculate $N\geq (1-\mathbb{P})q^{n^2-2n+1}$.
    
    For $t\in [1:n-P]$, let $p_t$ (respectively, $p_t'$) be the probability that $d_{L_1}\big(\mathrm{Comp}(\bm{V}_{[1:n],t}), \mathrm{Comp}(\bm{V}_{[1:n],t+\tilde{t}})\big)\leq 2$ (respectively, $d_{L_1}\big(\mathrm{Comp}(\bm{V}_{t,[1:n]}), \mathrm{Comp}(\bm{V}_{t+\tilde{t},[1:n]})\big)\leq 2$) for $\tilde{t}\in [1:P]$.
    Clearly, we have $p_t=p_t'=p_1$.
    Moreover, for $t\in [1:n-1]$ and $k\in [1:3]$, let $p_{t,k}$ denote the probability that $V_{\tilde{t},t}= V_{\tilde{t},t+1}$ for $\tilde{t}\in [(k-1)\ell+1:k\ell]\setminus \{t,t+1,(t+2 \mod{n})\}$.
    Clearly, we have $p_{t,k}\leq \frac{1}{q^{\ell-3}}$.
    Then by the union bound, we can calculate $\mathbb{P}\leq 2(n-2)p_1+\frac{3n}{q^{\ell-3}}$. 
    
    Let $M$ denote the maximum size of a set of $q$-ary sequences of length $n-2$ with identical composition. 
    By Equation (\ref{eq:M}), we get
    \begin{equation*}
        M\leq \sqrt{\frac{q^q}{(2\pi)^{q-1}(n-2)^{q-3}}} \cdot \frac{q^{n-2}}{n-2}.
    \end{equation*}
    \begin{itemize}
      \item When $q=2$, by Remark \ref{rmk:L_1}, we can compute $p_1\leq \left(\frac{5M}{q^{n-2}}\right)^P\leq \left(\frac{50}{\pi(n-2)}\right)^{\frac{P}{2}}$.
          In this case, we get $\mathbb{P}\leq \frac{3n}{2^{\ell-3}}+ \frac{2(50/\pi)^{P/2}}{(n-2)^{(P-2)/2}}$.
          Moreover, if $\ell\geq \log n+9$, $P\geq 8$, and $n\geq 58$, we have $\frac{50}{\pi(n-2)}\leq 1$ and can further compute $\mathbb{P}\leq \frac{3}{2^6}+\frac{2\cdot(50/\pi)^{4}}{56^{3}}\leq \frac{4}{5}$.
      \item When $q=3$, by Remark \ref{rmk:L_1}, we can compute $p_1\leq \left(\frac{40M}{q^{n-2}}\right)^P\leq \left(\frac{10800}{\pi^2(n-2)^2}\right)^{\frac{P}{2}}$.
          In this case, we get $\mathbb{P}\leq \frac{3n}{3^{\ell-3}}+ \frac{2\cdot(10800)^{P/2}}{\pi^P(n-2)^{P-1}}$.
          Moreover, if $\ell\geq \log n+5$, $P\geq 8$, and $n\geq 70$, we have $\frac{10800}{\pi^2(n-2)^2}\leq 1$ and can further compute $\mathbb{P}\leq \frac{1}{3}+\frac{2\cdot 10800^4}{\pi^8 \cdot 68^7}\leq \frac{4}{5}$.
      \item When $q\geq 4$, by Remark \ref{rmk:L_1}, we can compute $p_1\leq \left(\frac{q^4M}{q^{n-2}}\right)^P\leq \left(\frac{q^{q+8}}{(n-2)^{q-1}}\right)^{\frac{P}{2}}$.
          In this case, we get $\mathbb{P}\leq \frac{3n}{q^{\ell-3}}+ \frac{2q^{(q+8)P/2}}{(n-2)^{(q-1)P/2-1}}$.
          Moreover, if $\ell\geq \log n+4$, $P\geq 8$, and $n\geq q^{5}+2$, we have $\frac{q^{q+8}}{(n-2)^{q-1}}\leq 1$ and can further compute $\mathbb{P}\leq \frac{3}{4}+\frac{2\cdot q^{4q+32}}{q^{5(4q-5)}}=\frac{3}{4}+\frac{2}{q^{16q-57}}\leq \frac{3}{4}+\frac{2}{4^{7}}\leq \frac{4}{5}$.
    \end{itemize}
    Then the conclusion follows.
\end{proof}

With the help of Lemma \ref{lem:valid}, we are now prepared to calculate the redundancy of $\mathcal{C}_2(\boldsymbol{c},\boldsymbol{d},\boldsymbol{d}')$.

\begin{corollary}\label{cor:binary}
  Assume $n\geq 58$ if $q=2$,  $n\geq 70$ if $q\geq 3$, and $n\geq q^5+2$ if $n\geq 4$. Let $\ell=\log n+9$, $P=8$, and $q|(n-2)$.
  There exists a choice of parameters such that the redundancy of $\mathcal{C}_2(\boldsymbol{c},\boldsymbol{d},\boldsymbol{d}')$ (as defined in Construction \ref{constr2}) is at most $(2n-1)\log q+2\log n+4\log (18)+\log 5$, which is optimal up to an additive constant.
\end{corollary}

\begin{proof}
By considering all possible choices for $\boldsymbol{c}\in \Sigma_n^2$, $\boldsymbol{d}\in \Sigma_2^4$, and $\boldsymbol{d}'\in \Sigma_9^4$, we find that the total number of codes $\mathcal{C}_2(\boldsymbol{c},\boldsymbol{d},\boldsymbol{d}')$ is $(18)^4 n^2$, which forms a partition of valid arrays in $\Sigma_q^{n \times n}$. Consequently, by the pigeonhole principle, there exists at least one choice of $\boldsymbol{c}\in \Sigma_n^2$, $\boldsymbol{d}\in \Sigma_2^4$, and $\boldsymbol{d}'\in \Sigma_9^4$ such that the code $\mathcal{C}_2(\boldsymbol{c},\boldsymbol{d},\boldsymbol{d}')$ contains at least $\frac{q^{n^2-2n+1}}{5\cdot(18)^4\cdot n^2}$ codewords. Thus, the conclusion follows.
\end{proof}

\section{Construction of Optimal \texorpdfstring{$(t_r,t_c)$}{}-Criss-Cross Burst-Deletion Correcting Codes Over Arbitrary Alphabets}\label{sec:burst}

In one-dimensional sequences, it is well known that correcting a burst-deletion poses a significant challenge compared to correcting a single deletion error. 
An optimal construction for correcting a burst-deletion was introduced recently in \cite{Sun-25-IT-burst}, while codes capable of correcting a single deletion were reported in 1966 in \cite{Levenshtein-66-SPD-1D}.
In comparison, we will demonstrate that in two-dimensional arrays, the difficulty of correcting a $(t_r, t_c)$-criss-cross burst-deletion is nearly equivalent to that of correcting a $(1, 1)$-criss-cross deletion.  
  
\begin{definition}
  For a given array $\boldsymbol{X} \in \Sigma_q^{n \times n}$, let $\boldsymbol{X}^{(s_r, s_c)}$ be the subarray of $\boldsymbol{X}$ indexed by all rows whose indices are congruent to $s_r$ modulo $t_r$ and all columns whose indices are congruent to $s_c$ modulo $t_c$, for $s_r \in [1:t_r]$ and $s_c \in [1:t_c]$. For simplicity, we assume $t_r \mid n$ and $t_c \mid n$; then  
  
  \[  
    \boldsymbol{X}^{(s_r, s_c)} :=   
    \begin{bmatrix}  
       X_{s_r,s_c} & X_{s_r,s_c+t_c} & \cdots & X_{s_r, n-t_c+s_c} \\
    X_{s_r+t_r,s_c} & X_{s_r+t_r,s_c+t_c} & \cdots & X_{s_r+t_r, n-t_c+s_c} \\
    \vdots & \vdots & \ddots & \vdots \\
    X_{n-t_r+s_r,s_c} & X_{n-t_r+s_r,s_c+t_c} & \cdots & X_{n-t_r+s_r, n-t_c+s_c}
    \end{bmatrix}.  
    \] 
\end{definition} 

\begin{observation}
Observe that $\boldsymbol{X}^{(s_r, s_c)}$ suffers a $(1,1)$-criss-cross deletion for $s_r \in [1: t_r]$ and $s_c \in [1: t_c]$ when $\boldsymbol{X}$ experiences a $(t_r, t_c)$-criss-cross burst-deletion. Additionally, if $\boldsymbol{X}^{(1,1)}$ suffers a row deletion at position $i$ and a column deletion at position $j$, then the row deletion position in $\boldsymbol{X}^{(s_r, s_c)}$ is either $i$ or $i-1$, and the column deletion position in $\boldsymbol{X}^{(s_r, s_c)}$ is either $j$ or $j-1$ for $s_r \in [1: t_r]$ and $s_c \in [1: t_c]$.  
\end{observation}

To this end, we can encode $\boldsymbol{X}^{(1,1)}$ into some $(1,1)$-criss-cross deletion correcting code, with the additional constraint that decoding $\boldsymbol{X}^{(1,1)}$, which is distorted by a $(1,1)$-criss-cross deletion, can enable us to precisely identify the row and column deletion positions.   
Then, for $s_r \in [1: t_r]$ and $s_c \in [1: t_c]$ with $(s_r, s_c) \neq (1, 1)$, we require that $\boldsymbol{X}^{(s_r, s_c)}$ can correct a $(1,1)$-criss-cross deletion given a row interval of two containing the row deletion position and a column interval of length two containing the column deletion position. These requirements can be satisfied using the idea introduced in Construction \ref{constr2} and Remarks \ref{rmk:inversion} and \ref{rmk:row}.  
For completeness, we provide the entire construction of $(t_r, t_c)$-criss-cross burst-deletion correcting codes as follows.

\begin{construction}
  Let $P, \ell$ be positive integers satisfying $\frac{n}{t_r}\geq 3\ell$, $\frac{n}{t_c}\geq \ell$, and $\ell\geq P-1$. 
  For fixed sequences $\boldsymbol{c} = (c_1, c_2)$ with $c_1 \in \Sigma_{\frac{n}{t_r}}$ and $c_2 \in \Sigma_{\frac{n}{t_c}}$, $\boldsymbol{d} = (d_1, d_2, d_3,d_4) \in \Sigma_2^{4}$, $\boldsymbol{d}' = (d_1', d_2', d_3',d_4') \in \Sigma_{P+1}^{4}$, and $\boldsymbol{d}'' = (d_5'', d_6'', \cdots,d_{4t_rt_c}'') \in \Sigma_2^{4t_rt_c-4}$, we define the code $\mathcal{C}_3(\boldsymbol{c}, \boldsymbol{d},\boldsymbol{d}',\boldsymbol{d}'')$ in which each array $\boldsymbol{X}$ satisfies the following constraints:  
    \begin{itemize}  
        \item $\boldsymbol{X}^{(1,1)}$ belongs to the code $\mathcal{C}_2(\boldsymbol{c}, \boldsymbol{d},\bm{d}')$ defined in Construction \ref{constr2}, satisfying the additional constraint that any two consecutive rows are distinct; 
        \item For $s_r \in [1: t_r]$ and $s_c \in [1: t_c]$ with $(s_r, s_c) \neq (1, 1)$,  
        \begin{itemize}  
            \item $\boldsymbol{X}^{(s_r, s_c)}$ is $\ell$-weakly-valid;
            \item $\mathrm{Inv}\big(\mathrm{CIR}(\boldsymbol{X}_{[(k-1)\ell+1:k\ell], [1:\frac{n}{t_c}]}^{(t_r, t_c)})\big) \equiv d_{4t_c(s_r-1)+4(s_c-1)+k}'' \pmod{2}$ for $k \in [1:3]$;  
            \item $\mathrm{Inv}\big(\mathrm{RIR}(\boldsymbol{X}^{(t_r, t_c)})\big) \equiv d_{4t_c(s_r-1)+4s_c}'' \pmod{2}$.  
        \end{itemize}  
    \end{itemize}  
\end{construction}

\begin{remark}
Following a similar discussion in the proof of Theorem \ref{thm:binary}, one can demonstrate that $\mathcal{C}_3(\boldsymbol{c}, \boldsymbol{d},\boldsymbol{d}',\boldsymbol{d}'')$ is a $(t_r, t_c)$-criss-cross burst-deletion correcting code. Furthermore, analogous to the proof of Corollary \ref{cor:binary}, it can be shown that there exists a choice of parameters such that the redundancy of $\mathcal{C}_3(\boldsymbol{c}, \boldsymbol{d},\boldsymbol{d}'\boldsymbol{d}'')$ is at most $(t_r + t_c)n \log q + 2 \log n + O_{q, t_r, t_c}(1)$. Therefore, by Remark \ref{rmk:burst}, this construction is optimal (up to an additive constant) quantified by redundancy. 
\end{remark}

\begin{remark}
  When $q\geq 3$, we can also employ the code $\mathcal{C}_1(c,d)$ described in Construction \ref{constr1} to create optimal $(t_r, t_c)$-criss-cross burst-deletion correcting codes.
\end{remark}

\section{Conclusion}\label{sec:concl}
In this paper, we establish a sphere-packing type lower bound and a Gilbert-Varshamov type upper bound on the redundancy of $(t_r,t_c)$-criss-cross deletion correcting codes.
For the case where $(t_r,t_c)=(1,1)$, we present code constructions whose redundancy approaches the sphere-packing type lower bound, up to an additive constant.
Moreover, our methods can be used to construct optimal $(t_r,t_c)$-criss-cross burst-deletion correcting codes. 

Two intriguing open problems in two-dimensional deletion-corrections arise from this work:
\begin{itemize}
  \item Can our code constructions admit polynomial-time encoders while maintaining optimal redundancy? Current existential proofs, which rely on the pigeonhole principle, the probabilistic method, and other techniques, leave open the question of efficient encoding implementations.  
  \item Do there exist $(t_r,t_c)$-criss-cross (non-burst) deletion correcting codes achieving redundancy $(t_r + t_c)n\log q + (t_r + t_c)\log n + O_{q,t_r,t_c}(1)$ matching our sphere-packing lower bound?
\end{itemize}

\end{document}